\documentclass[journal,comsoc]{IEEEtran}
\usepackage{graphicx}
\usepackage{amssymb}
\usepackage{amsmath}
\usepackage{algorithmic}
\usepackage{algorithm}
\usepackage{array}
\usepackage{enumitem}
\usepackage{textcomp}
\usepackage{physics}
\usepackage{stfloats}
\usepackage{url}
\usepackage{hyperref}
\usepackage{verbatim}
\usepackage{booktabs}
\usepackage{graphicx}
\usepackage{tabularx}
\usepackage{cite}
\usepackage[short, c3]{optidef}
\usepackage[mathlines]{lineno}
\usepackage{color}
\usepackage{soul}
\usepackage{tabulary}
\usepackage{multirow}
\usepackage{pbox}
\usepackage{multicol}
\usepackage{lipsum}
\usepackage{amsthm}
\usepackage{relsize}
\usepackage{lipsum}
\usepackage{epstopdf}
\usepackage{mathtools}
\usepackage{calc}
\usepackage{makecell}
\usepackage{url}
\usepackage{pgfplots}\pgfplotsset{compat=1.18}
\usepackage[]{subfig}
\newtheorem{definition}{Definition}
\newtheorem{theorem}{Theorem}
\newtheorem{lemma}{Lemma}

\newtheorem{remark}{Remark}

\newcommand{\myrule}{\vspace{-0.6\baselineskip}\hrulefill\vspace{-0.1\baselineskip}}

\begin{document}
\title{\huge Fairness Designs for Load Balancing Optimization in Satellite-Cell-Free Massive MIMO Systems}
\author{Trinh Van Chien, Ngo Tran Anh Thu, Nguyen Hoang Lam, Hien Quoc Ngo, ~\IEEEmembership{Fellow,~IEEE}, Symeon Chatzinotas, ~\IEEEmembership{Fellow,~IEEE}, and Huynh Thi Thanh Binh,~\IEEEmembership{Member,~IEEE}\vspace{-5mm}
\thanks{Trinh Van Chien, Ngo Tran Anh Thu, Nguyen Hoang Lam, and Huynh Thi
Thanh Binh are with the School of Information and Communication Technology (SoICT), Hanoi University of Science and Technology (HUST), Vietnam (email: chientv@soict.hust.edu.vn, ngotrananhthu1632001@gmail.com, hoanglam02qh@gmail.com, binhht@soict.hust.edu.vn). Hien Quoc Ngo is with the School of Electronics, Electrical Engineering and Computer Science,
Queen's University Belfast, Belfast BT7 INN, United Kingdom (email: hien.ngo@qub.ac.uk). Symeon Chatzinotas is with the Interdisciplinary Centre for Security, Reliability and Trust (SnT), University of Luxembourg, L-1855 Luxembourg, Luxembourg (email: symeon.chatzinotas@uni.lu). (Corresponding author: Trinh Van Chien). 
}}

\markboth{Journal of \LaTeX\ Class Files,~Vol.~18, No.~9, September~2020}%
{How to Use the IEEEtran \LaTeX \ Templates}

\maketitle
\begin{abstract}
Space-ground communication systems are important in providing ubiquitous services in a large area. This paper considers the fairness designs under a load-balancing framework with heterogeneous receivers comprising access points (APs) and a satellite. We derive an ergodic throughput of each user in the uplink data transmission for an arbitrary association pattern and imperfect channel state information, followed by a closed-form expression with the maximum-ratio combining and rich scattering environments. We further formulate a generic fairness optimization problem, subject to the optimal association patterns for all the users. Despite the combinatorial structure, the global optimal solution to the association patterns can be obtained by an exhaustive search for small-scale networks with several APs and users. We design a low computational complexity algorithm for large-scale networks based on evolutionary computation that obtains good patterns in polynomial time. Specifically, the genetic algorithm (GA) is adapted to the discrete feasible region and the concrete fairness metrics. We extensively observe the fairness design problem by incorporating transmit power control and propose a hybrid genetic algorithm to address the problem.
Numerical results demonstrate that the association pattern to each user has a significant impact on the network throughput. Moreover, the proposed GA-based algorithm offers the same performance as an exhaustive search for small-scale networks, while it unveils interesting practical association patterns as the network dimensions go large. The load-balancing approach, combined with  power control factors, significantly enhances system performance compared to conventional schemes and configurations with fixed factors.
\end{abstract}


\vspace{-5mm}
\section{Introduction}
The transition to the sixth generation (6G) of terrestrial-space communications marks the forthcoming frontier in wireless network technologies, set to succeed the ubiquitous 5G. These new 6G networks are still in the initial stages of research and development, but they are expected to bring new ways of connecting people with faster data speeds, reduced latency, and increased capacity compared to their predecessors\cite{smida2023full,matthaiou2021road,mahmoud20216g,liu20246g}. In particular, 6G may enable data communication at unprecedented terabit per second speeds, achieve sub-millisecond latency, and offer the capacity to facilitate an extraordinary amount of connected devices within a compact square kilometer area \cite{zhang20236g}. 
A technological breakthrough supporting 6G is the implementation of advanced Massive MIMO systems, designed with numerous antennas utilizing linear beamforming techniques to manage connections within highly concentrated networks of devices\cite{zhang2023performance}. 
Moreover, the advent of cell-free Massive MIMO technology is expected to further advance spectral and energy efficiency by overcoming the conventional limitations imposed by cellular networks. The design of these networks is fundamentally load balancing, characterized by distributed transmitters cooperatively operating to provide dedicated service to individual users, thereby intensifying the focus on creating a customized network experience \cite{van2022space}. This emphasis on load balancing Cell-free Massive MIMO network propels the user experience to new heights, making possible a range of disruptive technological applications, especially in the domain of immersive extended reality; however, Cell-free Massive MIMO is most effective within restricted coverage areas such as urban \cite{ammar2021user}.

Because of its potential to offer widespread connectivity to numerous users over extensive areas, the field of satellite communications has experienced renewed interest \cite{schwarz2019mimo}, \cite{van2022user}. While geostationary (GEO) satellites can cover large regions, they come with high costs, shared bandwidth, and significant latency \cite{van2022user}. These drawbacks have prompted the development of non-geostationary (NGSO) satellites, such as those in low-Earth orbit (LEO), which are poised to revolutionize radio systems with the expectation of integration into future 6G networks \cite{zhao2024transfg}. NGSO satellites, orbiting at lower altitudes than GEO ones, provide benefits like reduced latency and tailored coverage for particular uses or isolated locations\cite{jia2021uplink}. Integrating satellites with cell-free Massive MIMO is recognized as a network architecture for advancing wireless communication systems \cite{van2022space}. Nonetheless, to the best of our knowledge, there is a lack of research exploring a load balancing approach for hybrid space-ground systems, particularly under the framework of coherent signal processing and resource management. 

Evolutionary algorithms (EAs) continue to draw interest because of their demonstrated capabilities in solving complex optimization problems that are prevalent in real-world scenarios. As future networks should be multi-layered systems with various integrated technologies, evolutionary strategies like the genetic algorithm (GA) have shown promise in radio resource management, attributed to their effectiveness in effective system designs \cite{van2024performance,chien2024active,binh2024efficient}. The GA operates on evolutionary biology principles, utilizing selection, mutation, and crossover processes to evolve populations toward high-quality solutions \cite{mitchell1998introduction}. The strength and adaptability of GAs make them particularly suitable for navigating complex configurations, allowing them to search through large solution spaces to identify superior solutions. Integrating diverse cutting-edge technologies in 6G networks introduces significantly complex optimization problems, especially in extensive network environments\cite{coello2007evolutionary}. In addition, the ability of GAs to scale makes them ideal for deployment in networks that must support an enormous number of devices with a wide array of service needs\cite{strinati20196g}. Notably, the user association problem is a form of combinatorial optimization, which complicates the direct application of the standard GA.

In this paper, a load balancing network formed by cell-free Massive MIMO and satellite is studied to improve spectral efficiency (SE) with different fairness criteria. Our main contributions are summarized as follows
\begin{itemize}
\item We derive the uplink ergodic SE of each user under arbitrary association patterns. The closed-form ergodic SE expression is further obtained as the maximum ratio combining (MRC) locally applied at the APs and satellite. 
\item We formulate a generic fairness problem to seek the optimal association between each user and the ground and space receivers. For small-scale networks, an exhaustive search through the combinatorial structure can obtain the globally optimal association. 
\item We propose an efficient algorithm for large-scale networks by exploiting the GA adapted to optimize the association sets under a load balancing framework. The main concept is that binary encoding can be easily adapted to represent the user association.  \textcolor{black}{To further enhance system performance, we integrate the power control into the load balancing optimization problem. We then develop a GA framework incorporating hybrid variables tailored to the optimization structure. } 
\item Numerical results demonstrate the critical roles of the satellite and APs in enhancing the ergodic SE. The GA-based association designs offer the same performance as an exhaustive search for small-scale networks and further efficiently optimize for large-scale networks. \textcolor{black}{The transmit power control leads to improved system performance by approximately 20\% compared to alternative schemes.}
\end{itemize}
The rest of this paper is organized as follows: Section~\ref{Sec:Sys} presents the integrated satellite-cell-free massive MIMO system along with the channel estimation process. In Section~\ref{Sec:Uplinkdata}, the uplink data transmission is described, and then the ergodic SE for every user is obtained in closed form. We formulate the generic fairness optimization problem under the load balancing topology and exploit an exhaustive search to obtain the global optimum in Section~\ref{Sec:Fairness}; besides, for networks with numerous users and APs, we propose GA-based association designs with high-quality solutions. \textcolor{black}{ Section~\ref{sec:Fairness Prob with PC} extends the fairness optimization framework by incorporating power control mechanisms and proposes a hybrid genetic algorithm to efficiently solve the problem.}
Section~\ref{Sec:NumRe} provides numerical results to validate the analytical SE and efficiency of our proposals. Finally, Section~\ref{Sec:Con} draws main conclusions.

\textit{Notation}: Matrices and vectors are denoted by the capital and lower bold letters, respectively. The matrix transpose and Hermitian  are denoted by $(\cdot)^T$ and $(\cdot)^H$, respectively. The Euclidean norm is denoted by $\| \cdot \|$. The expectation of a random variable is denoted by $\mathbb{E}\{\cdot\}$, and $\mathcal{CN}(\cdot, \cdot)$, $\mathcal{U}(\cdot, \cdot)$, and $\mathcal{B}(\cdot, \cdot)$ represent the circularly symmetric Gaussian, uniform, and Bernoulli distributions, respectively. Notations
$\land, \lor, \lnot, \oplus$ are logical operators, specifically conjunction (AND), disjunction (OR), negation (NOT), and exclusive OR (XOR).
\section{Load Balancing Integrated System, Pilot Training, and Channel Estimation} \label{Sec:Sys}
This section delves into space-ground communications, taking the practical challenges posed by imperfect channel state information and limited associations.
\subsection{System and Channel Models}
We consider an integrated space-ground network where an LEO satellite with $M$ antennas cooperates with $N$ single-antenna APs to serve $K$ single-antenna users. For load balancing purposes which enable the high spectral and energy efficiency, we assume that users can flexibly associate with the satellite and APs. To obtain a good spectral efficiency for each terrestrial user, it should hold that $M+N >> K$.
As illustrated in Fig.\ref{fig:NetworkModel}, the APs utilize optical fronthaul links, whereas the satellite connects to the ground station via a radio downlink (feeder link). The ground station then relays the uplinks signals from users to the CPU. We assume that both the optical fronthaul links and the feeder link has imperfect channel gains modeled by a complex Gaussian distribution, affecting the performance of both the pilot training and data transmission phases.\footnote{By utilizing the optical links, the fronthaul enables ideal gains and therefore the exact CPU's position is not necessarily stated due to the assumption on the Gaussian distribution. However, extending our framework to incorporate fronthaul-aware pilot allocation and user association strategies represents a valuable direct for future work.}
Although the propagation channels vary over the time and frequency plane, we adopt the quasi-static fading model where the propagation links are static and frequency-flat over each coherence intervals of $\tau_c$ symbols.  We assume the network operates in a fast-fading environment, and $K$ symbols in each coherence interval are dedicated to the pilot training phase. The remaining $\tau_c - K$ symbols are used for the uplink data transmission.  

\begin{figure}[t]
    \centering    \includegraphics[width=0.45\textwidth]{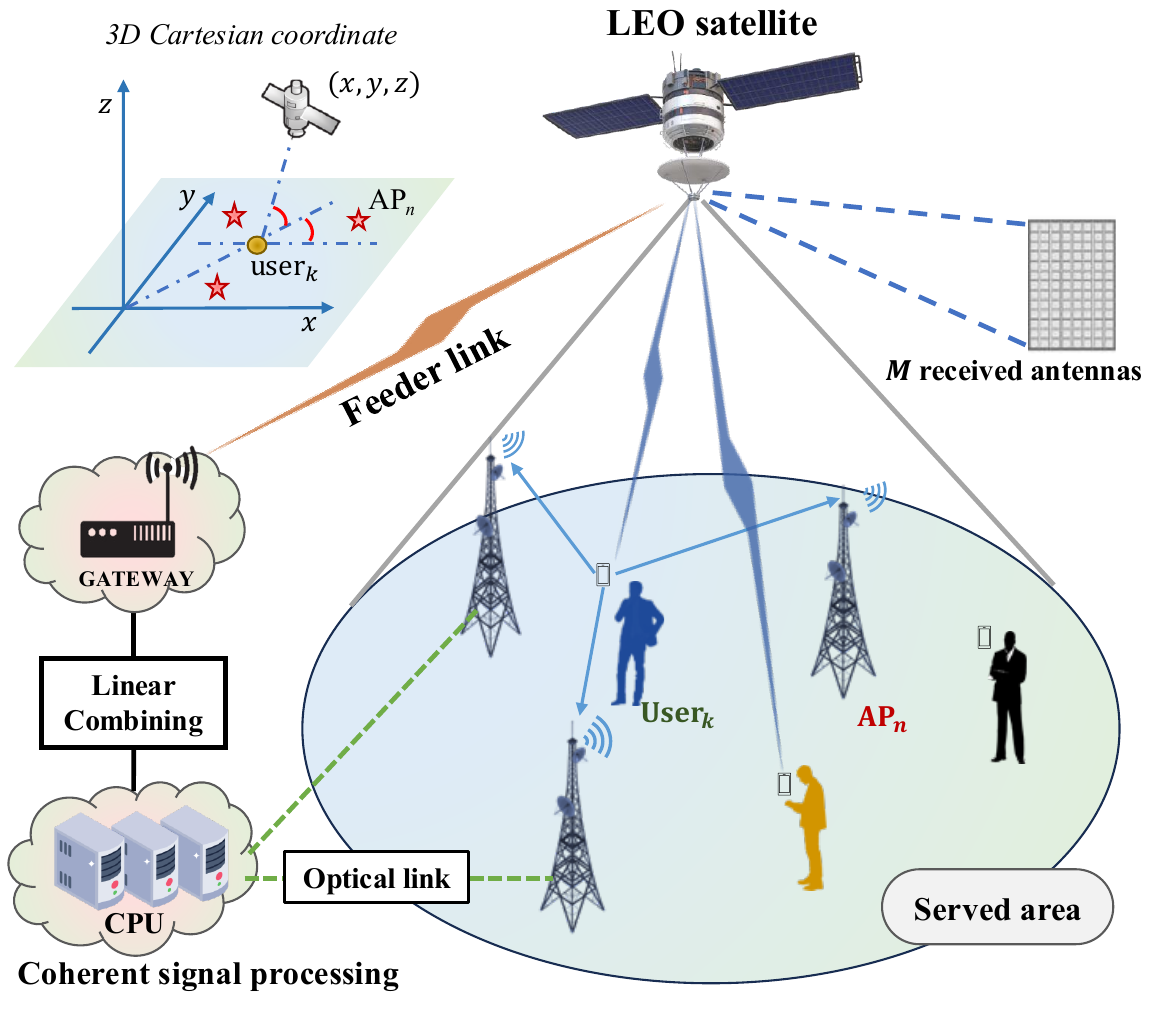}
    \caption{Illustration of a Satellite-Cell-Free MIMO network.}\vspace{-5mm}
    \label{fig:NetworkModel}
\end{figure}

Due to sharing the same time and frequency resource, each AP and satellite can simultaneously receive the transmitted signal from all the  users. Coherent processing allows us to define the terrestrial and/or non-terrestrial mode by optimizing the load balancing. For the load balancing design, let us denote $\tilde{\alpha}_k  \in \{ 0, 1\}$ the binary variable that establishes the association between user~$k$ and the satellite. If $\tilde{\alpha}_k = 1$, user~$k$ is served by the satellite. Otherwise, it is not served by the satellite. If user~$k$ is served by the satellite,  the space channel $\mathbf{h}_k \in \mathbb{C}^M$ is distributed as $\mathbf{h}_k \sim \mathcal{CN}(\bar{\mathbf{h}}_k, \mathbf{R}_k)$. Here, $\bar{\mathbf{h}}_k \in \mathbb{C}^M$ consists of the line of sight (LoS) components, while $\mathbf{R}_k \in \mathbb{C}^{M \times M}$ denotes the spatial correlation. For the user-APs association, we introduce the binary variables $\alpha_{k} \in \{ 0, 1\}$, where  $\alpha_{k} =1$ implies APs provide service to user~$k$, otherwise,  APs do not serve user~$k$. Besides, if user~$k$ is served by AP~$n$, $g_{nk} \in \mathbb{C}$ denotes the ground channel between AP $n$ and user $k$, which is distributed as $g_{nk} \sim \mathcal{CN}(0, \beta_{nk})$.\footnote{Uncorrelated Rayleigh fading is considered for terrestrial links where APs are deployed in rich scattering environments that are  aligned with measurement data. Inclusion of correlated Rayleigh to describe spatial correlation or Rician fading to characterize the LoS components are also of interest for particular propagation environments. These potential extensions are left for future work.}

\subsection{Uplink Pilot Training Phase}
During the pilot training phase in the uplink, all the $K$ users simultaneously transmit their pilot sequences to the receivers. More specifically, user~$k$ transmits its pilot sequence $\pmb{\psi}_k \in \mathcal{C}^{K}$ with $\| \pmb{\psi}_k \|_2 =1$ and $\pmb{\psi}_k^H \pmb{\psi}_{k'} = 0, \forall k \neq k'$. Then, the pilot signal received at AP~$n$, denoted by $\mathbf{y}_{pn} \in \mathbb{C}^K$,  and at the gateway, denoted by $\mathbf{Y}_{p}$, are respectively formulated as  
\begin{align}
&\mathbf{y}_{pn} = \sum\nolimits_{k=1}^K \sqrt{pK} g_{nk} \pmb{\psi}_k^H + \mathbf{n}_{pn}, \label{eq:ypm}\\
&\mathbf{Y}_{p} = \sum\nolimits_{k=1}^K \sqrt{pK}  \mathbf{h}_{k} \pmb{\psi}_k^H + \mathbf{N}_{p}, \label{eq:Yp}
\end{align}
where $p \geq 0$ is the transmit power allocated to each pilot symbol. The additive noise at AP~$n$ is denoted by $\mathbf{n}_{pn} \in \mathbb{C}^K$ and distributed as $\mathbf{n}_{pn} \sim \mathcal{CN}(\mathbf{0}, \sigma_a^2 \mathbf{I}_K)$. Besides, $\mathbf{N}_{p} \in \mathbb{C}^{M \times K}$ denotes the additive noise at the satellite whose elements are i.i.d. $\mathcal{CN}(0, \sigma_s^2)$ random variables. From \eqref{eq:ypm} and \eqref{eq:Yp}, the channel estimates and estimation errors are obtained by the minimum mean square error (MMSE) estimation \cite{Kay1993a}. 
\begin{lemma} \label{lemma:ChannelEst}
\textcolor{black}{If the MMSE estimation is exploited to estimate the channel $g_{nk}$, then the channel estimate $\hat{g}_{nk}$ is
\begin{equation} \label{eq:hatgnk}
\hat{g}_{nk} = \frac{pK \beta_{nk}}{ pK  \beta_{nk} + \sigma_{a}^2}  \mathbf{y}_{pn}^H \pmb{\psi}_k,
\end{equation}
 distributed as $\hat{g}_{nk} \sim \mathcal{CN}(0, \varrho_{nk} )$ and the variance $\varrho_{nk}$ as
\begin{equation} \label{eq:varrhonk}
\varrho_{nk} =  \frac{pK  \beta_{nk}^2}{ pK  \beta_{nk} + \sigma_{a}^2}.
\end{equation}
The estimation error, defined as $e_{mk} = g_{mk} - \hat{g}_{mk}$,  is distributes as $e_{mk} \sim \mathcal{CN}(0, \beta_{mk} - \varrho_{mk})$. Note that the channel estimate $\hat{g}_{mk}$ and the estimation error $e_{mk}$ are independent.}

If  the MMSE estimation is exploited to estimate the channel $\mathbf{h}_{k}$, the channel estimate $\hat{\mathbf{h}}_{k}$ is formulated as
\begin{equation} \label{eq:hatk}
\hat{\mathbf{h}}_k =  \bar{\mathbf{h}}_k + \sqrt{pK}\mathbf{R}_k \pmb{\Psi}_k ( \mathbf{Y}_p \pmb{\psi}_k -  \sqrt{pK} \bar{\mathbf{h}}_k ),
\end{equation}
where $\pmb{\Psi}_k = ( pK \mathbf{R}_k + \sigma_s^2 \mathbf{I}_M )^{-1} $. The channel estimate $\hat{\mathbf{h}}_k$ distributes as $\hat{\mathbf{h}}_k \sim \mathcal{CN}( \bar{\mathbf{h}}_k, p K  \mathbf{R}_k  \pmb{\Psi}_k  \mathbf{R}_k )$. Let us define the estimation error of the space link $\mathbf{e}_k = \mathbf{h}_k - \hat{\mathbf{h}}_k$, then it distributes as $\mathbf{e}_k \sim \mathcal{CN}(\mathbf{0},  \mathbf{R}_k - p K   \mathbf{R}_k  \pmb{\Psi}_k  \mathbf{R}_k )$. 
\end{lemma}
\begin{proof}
To estimate the channel $g_{nk}$, we first project the received pilot signal $\mathbf{y}_{pn}^H$ in \eqref{eq:ypm} into $\pmb{\psi}_k$ and perform the MMSE estimation as $\hat{g}_{nk} = \mathbb{E}\{g_{nk} | \mathbf{y}_{pn}^H \pmb{\psi}_k \}$ by noting that, in our case, the MMSE estimation and the linear MMSE estimation are the same. A similar methodology can be applied to the space link to obtain the channel estimate $\hat{\mathbf{h}}_k = \mathbb{E} \{ \mathbf{h}_k | \mathbf{Y}_p \pmb{\psi}_k, \overline{\mathbf{h}}_k \}$ by the MMSE estimation.
\end{proof}
The channel estimates and estimation errors are analytically obtained in Lemma~\ref{lemma:ChannelEst} as a function of the space and ground links. One can improve the channel estimation quality by carefully selecting the satellite and APs. The spatial correlation $\mathbf{R}_k$ shows the contributions to the channel estimates of the space links. The channel estimates in Lemma~\ref{lemma:ChannelEst} are applicable for the fast-fading models, whose channel errors are negligible under the limited pilot power and finite coherence time.

\section{Uplink Data Transmission and Ergodic Throughput Analysis} \label{Sec:Uplinkdata}
\subsection{Uplink Data Transmission}
In the uplink data transmission phase, all the $K$ users transmit signals to its served APs and/or the satellite. In particular, user~$k$ transmits a symbol $s_k$ with $\mathbb{E}\{ |s_k|^2 \} =1$ and the received signal at AP~$n$, denoted by $y_n \in \mathbb{C}$, is 
\begin{equation} \label{eq:ym}
 y_{n} = \sum\nolimits_{k=1}^K  \sqrt{p_k} g_{nk} s_k + n_n,
\end{equation}
where $p_k$ is the transmit power that user~$k$ and and $n_n \sim \mathcal{CN}(0, \sigma_n^2)$ is the additive noise.  Similarly, the received signal at the gateway of the satellite, denoted by $\mathbf{y} \in \mathbb{C}^M$, is 
 \begin{equation} \label{eq:y}
\mathbf{y} = \sum\nolimits_{k=1}^K  \sqrt{p}_k \mathbf{h}_k s_k + \mathbf{n},
\end{equation}
where $\mathbf{n} \sim \mathcal{CN}(\mathbf{0}, \sigma_s^2 \mathbf{I}_M)$ is the additive noise during the uplink data transmission. At the CPU, the signal sent from user~$k$ is decoded by the following combination
\begin{equation} \label{eq:hatsk}
\hat{\mathbf{s}}_k = \mathbf{w}_k^H \tilde{\alpha}_k\mathbf{y} + \alpha_{k}\sum\nolimits_{n=1}^N w_{nk}^\ast y_n,
\end{equation}
where $\mathbf{w}_k \in \mathbb{C}^M$ and $w_{nk} \in \mathbb{C}$ are the detection vectors utilized at the gateway and AP~$n$, $\forall n$, respectively.

\subsection{Ergodic Throughput Analysis}
In order to decode the transmitted data symbol $s_k$ sent by user~$k$, let us use the formulations of the received signals in \eqref{eq:ym} and \eqref{eq:y} into \eqref{eq:hatsk} and obtain
\begin{multline} \label{eq:shatk}
\hat{s}_k = \sum\nolimits_{k'=1}^K \sqrt{p_{k'}} \left(\tilde{\alpha}_{k'}  \mathbf{w}_{k}^H \mathbf{h}_{k'} + \sum\nolimits_{n=1}^N \tilde{\alpha}_{k'} w_{nk}^\ast g_{nk'} \right) s_{k'} \\ + \mathbf{w}_k^H \mathbf{n} + \alpha_{k}\sum\nolimits_{n=1}^N w_{nk}^\ast n_n,
\end{multline}
which is the superposition of the transmitted signals associated with the several receivers and additive noise.  In \eqref{eq:shatk}, we denote that a proper load balancing approach will mitigate mutual interference, while additive noise is accumulated from the satellite and all the APs. Let us introduce a new variable
\begin{equation} \label{eq:okkprime}
o_{kk'} = \tilde{\alpha}_{k'}  \mathbf{w}_{k}^H \mathbf{h}_{k'} + \sum\nolimits_{n=1}^N \tilde{\alpha}_{k'} w_{nk}^\ast g_{nk'}, 
\end{equation}
which stands for the overall channel and taking the association into account. Besides, let us denote the aggregated noise as
\begin{equation} \label{eq:ntildek}
\tilde{n}_k = \mathbf{w}_k^H \mathbf{n} + \sum\nolimits_{n=1}^N w_{nk}^\ast n_n.
\end{equation}
After that, by utilizing \eqref{eq:okkprime} and \eqref{eq:ntildek} into \eqref{eq:shatk}, the decoded signal of user~$k$ performed at the CPU is equivalent to as
\begin{equation} \label{eq:hatskv1}
\begin{split}
& \hat{s}_k = \sum\nolimits_{k'=1}^K \sqrt{p_k'} o_{kk'} s_{k'} +   \tilde{n}_k  =  \sqrt{p_k} \mathbb{E}\{o_{kk} \} s_{k} + \\
& \sqrt{p_k} ( o_{kk} - \mathbb{E}\{o_{kk} \} ) s_{k} +  \sum\nolimits_{k'=1, k' \neq k}^K \sqrt{p_{k'}} o_{kk'} s_{k'} +   \tilde{n}_k,
\end{split}
\end{equation}
where the first part in the second equation of \eqref{eq:hatskv1} represents the desired signal from user~$k$ with a deterministic channel gain; the second part denotes the beamforming uncertainty; and the remaining parts contain mutual interference and noise. By exploiting the use-and-then-forget channel capacity bounding technique, the uplink ergodic throughput of user~$k$ is
\begin{equation} \label{eq:Rkv1}
R_k = B\left(1 - K/\tau_c\right) \log_2 (1 + \mathrm{SINR}_k), \mbox{ [Mbps]},
\end{equation}
where $B$~[MHz] is the system bandwidth and the effective signal-to-interference-and-noise ratio (SINR) is given as in \eqref{eq:SINRk}. The uplink ergodic throughput in \eqref{eq:Rkv1} can be applied for an arbitrary channel model and detection vectors by numerically evaluating several expectations in the numerator and denominator of \eqref{eq:SINRk} over many different realizations of small-scale fading coefficients. However, these expectations are costly for large-scale networks with many APs and users. By exploiting the MRC technique, one can compute \eqref{eq:Rkv1} in closed form as shown in Theorem~\ref{Theorem:SE}.\footnote{In this paper, we can derive the exact closed-form solution of the uplink ergodic rate for MRC detection for an arbitrary $M$, $N$, and $K$. Furthermore, MRC is a scalable framework to expand the network with many APs and users. Developing a framework to approximately derive a closed-form expression of the uplink ergodic rate as  partial MMSE detection can, in fact, be developed. However, the approach would be different because an assumption $(M + N)/K \rightarrow \infty $ at a fixed rate should be used to align the closed-form expression of the uplink ergodic rate and Monte-Carlo simulations. Since this is a different approach, we would prefer to leave this interesting issue for our future research.}
\begin{figure*}[t]
\begin{equation} \label{eq:SINRk}
\mathrm{SINR}_k = \frac{p_k |\mathbb{E}\{o_{kk} \} |^2 }{\sum_{k'=1}^K p_{k'} \mathbb{E}\{ |o_{kk'}|^2 \} - p_k|\mathbb{E}\{o_{kk} \} |^2 + \mathbb{E}\{\|\mathbf{w}_k\|^2\} \sigma_s^2 + \sum_{n=1}^N \mathbb{E}\{|w|_{nk}^2\} \sigma_n^2}.
\end{equation}
\hrule
\end{figure*}
\begin{theorem} \label{Theorem:SE}
\textcolor{black}{If the MRC technique is deployed at the gateway and APs, the uplink ergodic throughput of user~$k$ is
\begin{equation}
R_k^{\mathrm{mrc}} = B \left(1 - K/\tau_c \right) \log_2 (1 + \mathrm{SINR}_k^{\mathrm{mrc}}), \mbox{[Mbps]},
\end{equation}
where the effective SINR is computed as
\begin{equation} \label{eq:SINRmrc}
\mathrm{SINR}_k^{\mathrm{mrc}} = \frac{p_k\left(\tilde{\alpha}_k\lVert \overline{\mathbf{h}}_k\rVert^2 +\tilde{\alpha}_kpK\operatorname{tr}(\pmb{\Theta}_k) + \alpha_k\sum_{n=1}^{N}\varrho_{nk}\right)^2}{\mathsf{MI}_k+\mathsf{NO}_k}
\end{equation}
The mutual interference, denoted as $\mathsf{MI}_k$, and the noise, denoted as $\mathsf{NO}_k$, are respectively defined as follows
\begin{equation}
    \begin{split}
        \mathsf{MI}_k = &\ \tilde{\alpha}_k \tilde{\alpha}_{k'} \Big( \sum\nolimits_{k' =1 , k' \neq k }^K p_{k'} \lvert \bar{\mathbf{h}}_{k}^H \bar{\mathbf{h}}_{k'}\lvert^2  + p K \sum\nolimits_{k' =1}^K p_{k'}    \bar{\mathbf{h}}_{k'}^H\pmb{\Theta}_k \bar{\mathbf{h}}_{k'} \\
&+ \sum\nolimits_{k' =1}^K p_{k'} \bar{\mathbf{h}}_{k}^H \mathbf{R}_{k'} \bar{\mathbf{h}}_{k}  + p K \sum\nolimits_{k' =1 }^K p_{k'}   \mathrm{tr}(\mathbf{R}_{k'}\pmb{\Theta}_k) \Big)  \\
& + \alpha_k\alpha_{k'}\sum\nolimits_{k' =1}^K \sum\nolimits_{n=1}^N p_{k'}  \varrho_{nk} \beta_{nk'}, \label{eq:MIk}
    \end{split}
\end{equation}
\begin{equation}
\mathsf{NO}_k = \tilde \alpha_k \left(\sigma_s^2 \lVert\bar{\mathbf{h}}_k\lVert^2 +  p K \sigma_s^2 \mathrm{tr}( \pmb{\Theta}_k ) \right) + \alpha_k\sigma_a^2 \sum\nolimits_{n=1}^N  \varrho_{nk}. \label{eq:NOk}
\end{equation}}
\end{theorem}
\begin{proof}
The proof is based on computing the expectations in \eqref{eq:SINRk} with the channel estimates and estimation errors in Lemma~\ref{lemma:ChannelEst}. The detailed proof is available in Appendix~\ref{Appendix:SE}.
\end{proof}
The numerator of the SINR in \eqref{eq:SINRmrc} demonstrates the contributions of the space link with both the LoS and NLoS components. The array gain from the satellite antennas and coherent combination scales up with the order of $\mathcal{O}((M +N)^2)$ as shown in the numerator of \eqref{eq:SINRmrc}. The mutual interference due to multiple access is expressed in \eqref{eq:MIk}, which scales up with the order of $\mathcal{O}(KM + NK)$. Moreover, the strength of additive noise is added from the APs and satellite with the order of $\mathcal{O}(M+N)$. Consequently, the optimal association obtained for each user will effectively reduce mutual interference and noise, leading to improved data throughput. In our network setting, the coherence time is sufficient long to deploy the orthogonal pilot signals, so the system does not suffer from pilot contamination. For the scenarios with the short coherent time and the pilot reuse, large-scale fading decoding can efficiently mitigate both coherent and non-coherent interference thanks to sharing the signaling through the backhaul among the satellite and APs. This potential extension will be left for future work.

\section{Fairness Optimization for Integrated System} \label{Sec:Fairness}
This section presents and analyzes a trio of widely considered optimization problems that highlight the benefits of a cooperation between space and ground links.\vspace{-5mm}
\subsection{Problem Formulation}
Due to the intricate interactions among user utility functions, simultaneous optimization seems elusive. Our primary objective is to optimize the aggregate utility function $f({\alpha_{k} }, {\tilde{\alpha}_k})$ to match the requirements of user~$k$ within the system, where $\alpha_{k}, \forall k,$ denote the association rules for ground links, and $\tilde{\alpha}_k, \forall k,$ represent those for space links\footnote{\textcolor{black}{
To maintain analytical tractability and focus on the influence of channel conditions and interference, our model assumes an infinite number of connections per AP and satellite.}}. With the associated notations, the generic fairness optimization problem is 
\begin{subequations} \label{Problem:MaxMinQoS}
    \begin{alignat}{2}
        & \underset{\{ \alpha_{k} \}, \{ \tilde{\alpha}_k \}}{\textrm{maximize }}
          && f (\{\alpha_{k} \}, \{ \tilde{\alpha}_k \}) \label{eq:Obj1} \\
        & \textrm{subject to } &&
		\alpha_{k} \in \{ 0,1 \}, \forall k, \\
        &&& \tilde{\alpha}_k \in \{0,1\}, \forall k. 
    \end{alignat}
\end{subequations}
In this paper, the three popular choices are typically considered for the system utility function $f (\{\alpha_{k} \}, \{ \tilde{\alpha}_k \})$ as
\begin{itemize}[]
\item[$i)$] Arithmetic mean utility:
\begin{equation}
f (\{\alpha_{k} \}, \{ \tilde{\alpha}_k \}) = \frac{1}{K}\sum\nolimits_{k=1}^K R_k (\{\alpha_{k} \}, \{ \tilde{\alpha}_k \}).
\end{equation}
\item[$ii)$] Geometric mean utility:
\begin{equation}
f (\{\alpha_{k} \}, \{ \tilde{\alpha}_k \}) = \left(\prod\nolimits_{k=1}^{K} R_k (\{\alpha_{k} \}, \{ \tilde{\alpha}_k \}) \right)^{\frac{1}{K}}.
\end{equation}
\item[$iii)$] Max-min fairness utility:
\begin{equation}
f (\{\alpha_{k} \}, \{ \tilde{\alpha}_k \}) = \min\limits_{1\leq k \leq K} R_k (\{\alpha_{k} \}, \{ \tilde{\alpha}_k \}),
\end{equation}
\end{itemize}
with $R_k (\{\alpha_{k} \}, \{ \tilde{\alpha}_k \})$ being defined in \eqref{eq:Rkv1}. 
We emphasize that the max-min fairness optimization provides an equal performance across the users. Contrarily, the arithmetic throughput approach disregards the experience of each individual, prioritizing the overall system utility instead. The geometric mean utility function provides a middle ground between the sum-rate and max-min utilities. It achieves a balance by ensuring that the system maintains a high overall spectral efficiency and provides a reasonable quality of service for each user. The transmit power of each user is fixed  so that the framework totally focuses on examining association patterns and their implications on system-level spectral efficiency and fairness. Despite this assumption, the optimization problem is NP-hard due to the binary nature of the association variables. Adding power control would significantly increase the complexity of the algorithm design since it is transformed into a mixed-integer non-convex problem and  require more sophisticated  approaches to seek for the solution. Jointly optimizing association and transmit power is, therefore, left for future work.\vspace{-5mm}

\subsection{Globally Optimal Solution to the Fairness Design Problem}
The globally optimal solution to the fairness optimization problem in \eqref{Problem:MaxMinQoS} can be obtained through an exhaustive search, as outlined in Algorithm~\ref{Alg1}. This method involves generating all possible solutions and evaluating their corresponding fairness levels, which are defined by the objective function. Once all possibilities have been generated, the globally optimal solution that maximizes the objective function is selected. However, the computational complexity of an exhaustive search grows exponentially with the number of available users and receivers, including APs and the satellite, resulting in a complexity of $\mathcal{O}(4^K)$. As a result, Algorithm~\ref{Alg1} becomes impractical for large-scale networks due to the limited computing power of practical hardware configurations. While an exhaustive search may still be feasible for small-scale networks with a limited number of users and APs, alternative methods that are more efficient in solving fairness design problems must be explored for large-scale networks.

\begin{algorithm}[!t]
\caption{Exhaustive Search Algorithm}

\begin{algorithmic}[1]\label{Alg1}
\ENSURE The optimal solution

\myrule

\STATE Initialize the current maximum solution as $\text{maxVal} \leftarrow -\infty$.

\STATE Set parameter $\text{D} = 2K$.

\FOR{$k \gets 2^{\text{D}} - 1$ \textbf{down to} $1$}
    \STATE Convert the index $k$ to its $\text{D}$ - dimension binary representation, then convert to a vector with $2K$ dimensions similar to [\ref{bga_initial}].
    
    \STATE Using the created vector, evaluate the fitness value $F$.

    \IF{$F > \text{maxVal}$}
        \STATE Update $\text{maxVal} \leftarrow F$.
    \ENDIF
\ENDFOR

\RETURN $\text{maxVal}$ as the best solution found.
\end{algorithmic}
\end{algorithm}

\subsection{Binary-Coded Genetic Algorithm (BCGA)}
We introduce a low computational complexity algorithm based on GA associated with an individual encoding to address the three fairness optimization problems.
\begin{algorithm}
\caption{Binary-Coded Genetic Algorithm}
\begin{algorithmic} [1] \label{Alg:Integer-coded}
\REQUIRE Large-scale fading coefficient, bandwidth, number of subcarriers in an (OFDM) symbol, carrier frequency.
\ENSURE The sub-optimal solution.

\myrule

{$\setminus\setminus$ \texttt{Initialization}}

\STATE Generate initial population $\mathcal{S}_p$  [\ref{bga_initial}].
\STATE Set max generations $S_{MAX}$, initialize $S \leftarrow 1$.
\STATE Set parameters: crossover rate $p_c$, mutation rate $p_m$, offspring count $n_c = 2 \lfloor p_c \frac{Q}{2} \rfloor$, mutants count $n_m = \lfloor p_m Q \rfloor$.

\WHILE{$S \leq S_{MAX}$}
    \STATE Initialize next generation $\mathcal{S}_{p'} = \varnothing$.
    
    {$\setminus\setminus$ \texttt{Crossover}}
    \FOR{$k_c=1$ to $n_c/2$}
        \STATE Select two distinct parents from $\mathcal{S}_p$.
        \STATE Generate two offspring via [\ref{bga_crossover}], add to $\mathcal{S}_c$.
    \ENDFOR

    \FOR{$k_m=1$ to $n_m$}
        \STATE Select an individual from $\mathcal{S}_c$.
        \STATE Generate mutant via [\ref{bga_mutation}], add to $\mathcal{S}_m$.
    \ENDFOR
    
    {$\setminus\setminus$ \texttt{Selection}}

    \STATE Evaluate fitness of $\mathcal{S}_{p'} = \mathcal{S}_c \cup \mathcal{S}_m$.
    \STATE Select $Q$ individuals from $\mathcal{S}_p \cup \mathcal{S}_{p'}$ for next $\mathcal{S}_p$.
    
    \STATE $S \leftarrow S + 1$.
\ENDWHILE

\RETURN Best solution found.

\end{algorithmic}

\end{algorithm}
\subsubsection{Solution representation and Initialization} \label{bga_initial}
The population $\mathcal{S}_p$ in the BCGA represents a set of $Q$ potential solutions. The $i$-solution, $\pmb{\phi}_i=\{\pmb{\phi}_{i1},\pmb{\phi}_{i2},\ldots,\pmb{\phi}_{iK}\}$ with $i\in \{1,\ldots,Q\}$, represents a vector that is made up of $K$ segments, each having 2 dimensions. Thus, the total dimensions of the vector are $2K$. Each component of these vectors, referred to as genes, are binary, meaning they can be either 0 or 1. Specifically, the vector segment $\pmb{\phi}_{ij}$, being  two-dimensional denoted by $\pmb{\phi}_{ij}=\{{\phi}_{ij1},{\phi}_{ij2}\}$ as
$z=1$, $\phi_{ijz}$ indicates {the connectivity of user}~$j$ and APs, and $z=2$, $\phi_{ijz}$ indicates {the connectivity of user}~$j$ and the satellite.
In detail, with a value of one signifying a connection and zero indicating no connection. During the initialization phase of the algorithm, a special individual is created with all the $2K$ genes set to one, representing a scenario where the users are served by both the satellite and $N$ APs. The remaining $Q-1$ individuals are randomly assigned values as $\phi_{ijz} = 0$ if $u < 0.5$ with $i\in \{1,\ldots,Q-1\}, j \in \{1,\ldots,K\}, z\in\{1,2\}$, and $u \sim \mathcal{U}([0,1])$. Otherwise,  $\phi_{ijz} = 1$.

\subsubsection{Binary crossover operator} \label{bga_crossover}

Instead of limiting to one type of the crossover method, we incorporate three different crossover techniques, which are applied based on dynamically adjusted probabilities $\varepsilon_1$ and $\varepsilon_2$ that reflect their effectiveness during the search process. We create two children solutions $\pmb{\phi}_{c^1}$ and $\pmb{\phi}_{c^2}$ to form crossover offspring population $\mathcal{S}_c$ from the two parent solutions $\pmb{\phi}_{p^1}$ and $\pmb{\phi}_{p^1}$ as

\begin{equation}
\begin{cases}
     &\pmb{\phi}_{c^1} = (\mathbf{c}_{\text{mask}} \land \pmb{\phi}_{p^1}) \lor (\lnot \mathbf{c}_{\text{mask}} \land \pmb{\phi}_{p^2}),\\
     &\pmb{\phi}_{c^2} = (\mathbf{c}_{\text{mask}} \land \pmb{\phi}_{p^2}) \lor (\lnot \mathbf{c}_{\text{mask}} \land \pmb{\phi}_{p^1}),
\end{cases}
\end{equation}
and $\mathbf{c}_{\text{mask}}$ is especially designed for particular crossover. In particular, it holds that
\begin{itemize}
    \item One-point mask crossover with the probability of $\varepsilon_1$: $\mathbf{c}_{\text{mask}}[i] = 
        0, \text{if } i < \mathsf{cp},$ and $\mathbf{c}_{\text{mask}}[i] =
        1, \text{otherwise},$
where $\mathsf{cp}$ is the crossover point and $\mathsf{cp} \sim \mathcal{U}(K(N+1))$. 
    \item Two-point mask crossover with the probability of $\varepsilon_2$: $\mathbf{c}_{\text{mask}}[i] = 0,  \text{if } \mathsf{cp}_1 \leq i < \mathsf{cp}_2,$ and $ \mathbf{c}_{\text{mask}}[i] = 1,  \text{otherwise,}$
    where $\mathsf{cp}_1$ and $\mathsf{cp}_2$ are two pivotal crossover points. These points are dependently selected: $\mathsf{cp}_1 \sim \mathcal{U}(1,N(K+1) - 1)$ and $\mathsf{cp}_2 \sim \mathcal{U}(\mathsf{cp}_1+1,N(K+1))$.
    \item Uniform mask crossover with the probability of $1 - \varepsilon_1 - \varepsilon_2$: $
     \mathbf{c}_{\text{mask}} \sim \mathcal{B}(2K,0.5).$
Each $\mathbf{\text{mask}}[i]$ is either 0 or 1, chosen randomly with a probability of 0.5. The mask $\mathbf{c}_{\text{mask}}$ in uniform crossover is not determined by specific crossover points. Instead, it is generated randomly for each gene position across the entire length of the chromosome.
\end{itemize}

\subsubsection{Bitwise mutation operator}\label{bga_mutation}
Following the mutation process shown in pseudocode of Algorithm \ref{Alg:Integer-coded}, based on the probability $p_m$, the parent vector $\pmb{\phi}_c$ is selected from the current population $\mathcal{S}_p$ and undergoes mutation to produce the trial solution $\pmb{\phi}_m$ to form mutated offspring population $\mathcal{S}_m$ as


\begin{equation}
\pmb{\phi}_{m} = \pmb{\phi}_{c} \oplus \mathbf{M}(n_{\text{mutate}},2K),
\end{equation}
where $n_{\text{mutate}}$ is a set of indices randomly chosen from 1 to $2K$ and create a binary mask $\mathbf{M}(n_{\text{mutate}},2K)$ of length $2K$ where the selected indices are set to 1 (indicating mutation), and others are 0. Besides, the mutation probability $p_m$ is a critical parameter in this process. It determines the probability of each solution undergoing mutation. Setting a higher mutation probability can increase the frequency of mutations, allowing the algorithm to explore a broader range of genetic variations. This can be advantageous for escaping local optima and enhancing the genetic diversity within the population. However, if the mutation probability is too high, it might lead to premature convergence, where the algorithm settles too quickly on less optimal solutions. A low mutation probability may slow down the evolution, potentially causing the algorithm to stagnate in local optima due to inadequate exploration.
\subsubsection{Survival selection operator}\label{alg:BCGA-selection}
After generating the offspring population $\mathcal{S}_{p'} = \mathcal{S}_c \cup \mathcal{S}_m$, using the above crossover and mutation, the parent population $\mathcal{S}_p$ is combined with to form the merged population $\mathcal{S}_{\text{merged}}$ as:

\begin{equation}
\mathcal{S}_{\text{merged}} = \mathcal{S}_p \cup \mathcal{S}_{p'} = \mathcal{S}_p \cup \mathcal{S}_c \cup \mathcal{S}_m,
\end{equation}
then sort the merged population $\mathcal{S}_{\mathrm{merged}}$ based on the  fitness value and select the best $Q$ individuals to the new generation as
\begin{equation}\label{eq:sort_bcga}
\mathcal{S}_p = \texttt{argsort}(\texttt{value}(\mathcal{S}_{\text{merged}})[1:Q],
\end{equation}
the best solution is selected for the next generation, and the equation \eqref{eq:sort_bcga} guarantees that the objective function of \eqref{Problem:MaxMinQoS} does not decrease across generations.

\subsubsection{Convergence and Computational Complexity}

To analyze the computational complexity of the proposal described in Algorithm~\ref{Alg:Integer-coded} to solve the problem~(\ref{Problem:MaxMinQoS}), we break down the computational effort required by its different components. An initialization of the population requires the complexity of the order of $\mathcal{O}(QK)$. For each generation, if $\bar\varphi$ denotes the count of chromosomes involved in crossovers, the computational load for the crossover operation amounts to the order of $\mathcal{O}(\bar\varphi K)$. The mutation process operates at a basic computational complexity $\mathcal{O}(1)$, and the selection phase requires $\mathcal{O}(Q\log(Q))$. Combining these costs, the total computational complexity of Algorithm~\ref{Alg:Integer-coded} is in the order of $\mathcal{O}(QK + S(\bar\varphi K + Q\log(Q)))$, where $Q$ is the size of the population, $S$ represents the number of generations, and $K$ denotes the number of users. These complexity formulas show that the algorithm's performance depends on the number of users $K$ and the GA parameters $Q$ and $S$, but not on the number of antennas $M$, which is an important advantage of our approach.
Regarding the convergence, we define the expected first hitting time as the average number of generations needed for Algorithm~\ref{Alg:Integer-coded} to achieve a fixed solution as shown in  Theorem~\ref{Theorem:EFHT}.

\begin{theorem}
\label{Theorem:EFHT}
From an initial population, solving the fairness design in  problem \eqref{Problem:MaxMinQoS} using the BCGA with bitwise mutation at a specified probability $p_m$, and assuming the population size matches the solution size, the expected first hitting time is constrained by the following bound:
\begin{equation}
\mathbb{E}\{\tau \} \geq c(1-p_m)^{-2K}K^{-1},
\end{equation}
where $\tau$ signifies the number of generations until the best solution is reached, $c$ is a positive constant, and $p_m \in (0, 0.5]$ in our experiments.
 \end{theorem}

\begin{proof}
\textcolor{black}{The proof is completed by formulating the procedure of the BCGA as a Markov chain and investigating its features. The detailed proof is available in Appendix~\ref{Appendix:EFHT}.}
 \end{proof}
Based on Theorem~\ref{Theorem:EFHT}, our BCGA algorithm guarantees convergence to the optimum from an initial population as the number of generations increases, ensuring the reliability of our optimization framework even in complex network scenarios.
 
\begin{remark}
The optimization problem \eqref{Problem:MaxMinQoS} is NP-hard and complicated due to the network setting with both the terrestrial and nonterrestrial links. The combinatorial complexity makes it challenging to solve, especially for large-scale networks. That is the main reason why intelligent computation  algorithms as GA are proposed since the quality of the solution is updated without computing the first and second derivatives of the objective function and constraints. The extension of our framework to support multi-antennas APs is straightforward but non-trivial, which may change the mathematical analysis of network performance and optimization structure. We therefore leave this potential extension for future work.\vspace{-5mm}
 \end{remark}

\textcolor{black}{
\section{Fairness Optimization with Joint Data Power Control and Load Balancing}\label{sec:Fairness Prob with PC}
\subsection{Problem Formulation}
To further enhance the network performance and ensure fairness among users, we extend the association optimization framework by incorporating transmit power control. In this enhanced formulation, each user’s transmit power is constrained within a practical range $[0, P_{\max,k}]$, where $P_{\max,k}$ denotes the maximum allowable transmit power of user~$k$. The problem~(\ref{Problem:MaxMinQoS}) with  power control optimization is given by 
\begin{subequations} \label{Problem:powercontrol}
    \begin{alignat}{2}
        & \underset{\{ \alpha_{k} \}, \{ \tilde{\alpha}_k \}}{\textrm{maximize }}
          && f (\{\alpha_{k} \}, \{ \tilde{\alpha}_k \}) \label{eq:Obj2} \\
        & \textrm{subject to } &&
		\alpha_{k},\tilde{\alpha}_k \in \{ 0,1 \}, \forall k, \\
        &&& 0 \leq p_k \leq P_{\max,k}, \forall k.
    \end{alignat}
\end{subequations}
To normalize the power allocation variable, we define a new variable \(\xi_k\) such that 
$p_k = \xi_k P_{\max,k}, \text{ where }  0 \leq \xi_k \leq 1.$
Consequently, the problem \eqref{Problem:powercontrol} becomes
\begin{subequations} \label{Problem:powercontrol_new}
    \begin{alignat}{2}
        & \underset{\{ \alpha_{k} \}, \{ \tilde{\alpha}_k \},\{\xi_k\}}{\textrm{maximize }}
          && f (\{\alpha_{k} \}, \{ \tilde{\alpha}_k \}, \{\xi_k\}) \label{eq:Obj3} \\
        & \textrm{subject to } &&
		\alpha_{k}, \tilde{\alpha}_k \in \{ 0,1 \}, \forall k, \\
        &&& \xi_k \in [0, 1], \forall k.
    \end{alignat}
\end{subequations}
which makes the optimization problem more tractable, especially when employing optimization algorithms.
\begin{lemma} Let $\{p_k\}_{k=1}^K$ denote the original power allocation variables with individual constraints $0 \leq p_k \leq P_{\max,k}$. Define the normalized variables $\xi_k \triangleq \frac{p_k}{P_{\max,k}}$ such that $\xi_k \in [0,1]$ for all $k$. Suppose the system performance metric is modeled by a utility function $f(\mathbf{p}) = f(p_1, \dots, p_K),$
where $f$ is quasi-concave \cite{van2022space} over the box-constrained domain $\mathcal{P} = \prod_{k=1}^K [0, P_{\max,k}]$. Define the transformed utility function in the normalized domain as $\tilde{f}(\boldsymbol{\xi}) \triangleq f(P_{\max,1} \xi_1, \dots, P_{\max,K} \xi_K).$
Then, $\tilde{f}(\boldsymbol{\xi})$ is quasi-concave over the unit cube $\mathcal{B} = [0,1]^K$. Furthermore, if $f$ is strictly quasi-concave and continuously differentiable, then so is $\tilde{f}$, and the set of global maximizers of $f$ and $\tilde{f}$ correspond via a one-to-one affine mapping.
\end{lemma}
\begin{proof}
Let $\boldsymbol{\xi}^{(1)}, \boldsymbol{\xi}^{(2)} \in [0,1]^K$ and consider any $\lambda \in [0,1]$. Define the convex combination $\boldsymbol{\xi}^{(\lambda)} = \lambda \boldsymbol{\xi}^{(1)} + (1 - \lambda)\boldsymbol{\xi}^{(2)}$.
Under the affine transformation $p_k = \xi_k P_{\max,k}$, this implies $\mathbf{p}^{(\lambda)} = \lambda \mathbf{p}^{(1)} + (1 - \lambda)\mathbf{p}^{(2)},$ where $\mathbf{p}^{(i)} = (P_{\max,1} \xi_1^{(i)}, \dots, P_{\max,K} \xi_K^{(i)})$ for $i = 1,2$.
Since $f(\mathbf{p})$ is quasi-concave, we have:
$f(\mathbf{p}^{(\lambda)}) \geq \min\left\{ f(\mathbf{p}^{(1)}), f(\mathbf{p}^{(2)}) \right\}.$
Therefore, $\tilde{f}(\boldsymbol{\xi}^{(\lambda)}) = f(\mathbf{p}^{(\lambda)}) \geq \min\left\{ \tilde{f}(\boldsymbol{\xi}^{(1)}), \tilde{f}(\boldsymbol{\xi}^{(2)}) \right\},$ which proves that $\tilde{f}$ is quasi-concave on $[0,1]^K$. When $f$ is strictly quasi-concave and continuously differentiable, $\tilde{f}$ similarly exhibits these traits through the one-to-one affine transformation. The global maximizers correspond directly due to the invertibility of the mapping $p_k = P_{\max,k} \xi_k$.
\end{proof}
\subsection{Hybrid Genetic Algorithm (HGA)}
We propose a low-complexity Hybrid Genetic Algorithm (HGA) framework, incorporating a mixed-variable encoding scheme, to efficiently tackle the fairness-driven power control and user association optimization problems. In detail, the updated $i$-solution, denoted as $\pmb{\gamma}_i = \{\pmb{\gamma}_{i1}, \pmb{\gamma}_{i2}, \ldots, \pmb{\gamma}_{iK}\}$ for $i \in \{1, \ldots, Q\}$, is structured as a composite vector consisting of $K$ segments. Each segment comprises three core dimensions, augmented by an additional dimension to incorporate power control constraints. The current representation is hybrid in nature: one component comprises binary-coded vectors initialized as described in~\ref{bga_initial}, while the other consists of real-valued vectors.
For the binary components, the optimization procedure follows Algorithm~\ref{Alg:Integer-coded}, which addresses discrete association variables. Meanwhile, the real-valued segments are optimized using a Real-Coded Genetic Algorithm framework, following the methodologies i.e. Simulated Binary Crossover and Parameter-based Mutation based on~\cite{sbcf,parameter_real}.}
\textcolor{black}{\subsubsection{Simulated Binary Crossover (SBX)}\label{alg:SBX}
For the real-valued power allocation variables $\{\xi_k\}$, we employ SBX to generate offspring solutions. Given two parent solutions $\xi^{p1}_k$ and $\xi^{p2}_k$ for user $k$, SBX produces two offspring $\xi^{c1}_k$ and $\xi^{c2}_k$ through the following procedure: $i)$
Compute boundary parameters within the normalized bounds $[0,1]$: $\varepsilon_1 = 1 + \frac{2\xi^{p1}_k}{\xi^{p2}_k-\xi^{p1}_k},$ and     $\varepsilon_2 = 1 + \frac{2(1-\xi^{p2}_k)}{\xi^{p2}_k-\xi^{p1}_k},$ with assuming $\xi^{p1}_k < \xi^{p2}_k$. $ii)$ Then calculate: $\vartheta = 2 - \varepsilon^{-(\eta_c +1)}$. The distribution index $\eta_c$ controls the proximity of offspring to their parents. Based on a random parameter $\mu \sim \mathcal{U}(0,1)$, determine the spread factor $\bar{\varepsilon}$ using polynomial probability distribution as $\bar{\varepsilon}=(\vartheta \mu)^{\frac{1}{(\eta_c+1)}} , \text{if } \mu \leq \frac{1}{\vartheta},$ and otherwise $\bar{\varepsilon}=\left(\frac{1}{2- \vartheta \mu}\right)^{\frac{1}{(\eta_c+1)}}.$
$iii)$ Generate offspring power allocation variables
\begin{align}
&  \xi^{c1}_k = 0.5\left((\xi^{p1}_k+\xi^{p2}_k)-\bar{\varepsilon}|\xi^{p2}_k-\xi^{p1}_k|\right), \\
& \xi^{c2}_k = 0.5\left((\xi^{p1}_k+\xi^{p2}_k)+\bar{\varepsilon}|\xi^{p2}_k-\xi^{p1}_k|\right).
\end{align}
\subsubsection{Parameter-based Mutation}\label{alg:parameter_mutation}
For the normalized power variables $\xi_k \in [0,1]$, polynomial mutation generates a perturbed solution $\xi'_k$ in the neighborhood of the original $\xi_k$ as follows: $i)$ generate random parameter $\mu \sim \mathcal{U}(0,1)$; $ii)$ compute the perturbation factor $\bar{\delta}=\left(2\mu+(1-2\mu)(1-\delta_k)^{\eta_m+1}\right)^{\frac{1}{\eta_m+1}}-1 , \text{if } \mu \leq 0.5$ and $\bar{\delta}=1 - \left(2(1-\mu)+2(\mu-0.5)(1-\delta_k)^{\eta_m+1}\right)^{\frac{1}{\eta_m+1}}, \text{ otherwise.}$ Where $\delta_k = \min[\xi_k, (1-\xi_k)]$ ensures the mutated solution remains within $[0,1]$ and the distribution index $\eta_m$ controls mutation intensity; and 
$iii)$ generate the mutated power allocation
    \begin{equation}
        \xi'_k = \xi_k + \bar{\delta}.
    \end{equation} 
\subsection{Convergence and Computational Complexity for HGA}
To analyze the computational complexity of the proposed HGA solving problem~(\ref{Problem:powercontrol_new}), we examine its hybrid components. Population initialization requires $\mathcal{O}(QK)$ complexity. For each generation, binary crossover operations require $\mathcal{O}(\bar{\varphi}_b K)$, real-valued crossover (SBX) requires $\mathcal{O}(\bar{\varphi}_r K)$, real-valued mutation requires $\mathcal{O}(\bar{\upsilon}_r K)$, and selection requires $\mathcal{O}(Q\log(Q))$. The total computational complexity is $\mathcal{O}(QK + S((\bar{\varphi}_b+\bar{\varphi}_r+\bar{\upsilon}_r)K + Q\log(Q)))$, where $\bar{\varphi}_b$, $\bar{\varphi}_r$, and $\bar{\upsilon}_r$ represent the number of individuals participating in binary crossover, real-number crossover, and real-number mutation per generation, respectively.
For convergence analysis, we must account for both discrete (association variables) and continuous (power control variables) components.}
\textcolor{black}{\begin{theorem}
\label{Theorem:EFHT_HGA}
Consider the HGA solving problem \eqref{Problem:powercontrol_new} with hybrid solution space $\mathcal{Y} = \{0,1\}^{2K} \times [0,1]^K$, where the optimal solution set is $\mathcal{Y}^* \subset \mathcal{Y}$. Let $p_m, \eta_m \in (0, 0.5]$ be the binary mutation probability and the polynomial mutation distribution index, respectively. The expected of EFHT $\tau$ satisfies
\begin{equation}
\mathbb{E}\{\tau\} \geq \tilde{c} (1 - p_m)^{-2K}  (1-\eta_m)^{-K} K^{-1},
\end{equation}
where $\tilde{c} > 0$ is a constant, and $\tau$ is the number of generations until the optimum is found.
\end{theorem}
\begin{proof}
The proof is completed by formulating the procedure
of the HGA as a Markov chain and investigating its features. The detailed proof is available in Appendix~\ref{Appendix:Proof_EFHTofHGA}.
\end{proof}
\begin{remark}
Theorem~\ref{Theorem:EFHT_HGA} reveals that HGA convergence is limited by the slower of two components binary and real values.
\end{remark}}

 \begin{algorithm}
\caption{\textcolor{black}{Hybrid Genetic Algorithm (HGA)}}
\begin{algorithmic} [1] \label{alg:HybridGA}
{\color{black} \REQUIRE Large-scale fading coefficient, bandwidth, number
of subcarriers in an (OFDM) symbol, carrier frequency.
\ENSURE The sub-optimal olution.

\myrule

{$\setminus\setminus$ \texttt{Initialization}}

\STATE Generate initial population $\mathcal{G}_p$ with hybrid individuals.
\STATE Set max generations $S_{MAX}$, initialize $S \leftarrow 1$.
\STATE Set crossover and mutation parameters: $p_c$, $p_m$, $\eta_n$, $\eta_m$.

\WHILE{$S \leq S_{MAX}$}
    \STATE Initialize next generation $\mathcal{G}_{p'} = \varnothing$.
    
    {$\setminus\setminus$ \texttt{Crossover}}
    \FOR{$k_c=1$ to $2\lfloor(p_c+\eta_c)Q/4\rfloor$}
        \STATE Select two parents from $\mathcal{G}_p$.
        \STATE Generate two offspring. Binary vector follows as in~[\ref{bga_crossover}]; meanwhile, the real-valued vector follows as in [\ref{alg:SBX}], add to $\mathcal{G}_c$.
    \ENDFOR
    
    {$\setminus\setminus$ \texttt{Mutation}}
    \FOR{$k_m=1$ to $\lfloor(p_m+\eta_m)Q/2\rfloor$}
        \STATE Select  an individual from from $\mathcal{G}_c$.
        \STATE Generate mutant. Binary vector follows as in~[\ref{bga_mutation}]; meanwhile, the real-valued vector follows as in [\ref{alg:parameter_mutation}], add to $\mathcal{G}_m$.
    \ENDFOR
    
    {$\setminus\setminus$ \texttt{Selection}}
    \STATE Evaluate fitness of each individual in $\mathcal{G}_{p'} = \mathcal{G}_b \cup \mathcal{G}_m$.
    \STATE Select $Q$ best individuals from $\mathcal{G}_p \cup \mathcal{G}_{p'}$.
    \STATE $S \leftarrow S + 1$.
\ENDWHILE
\RETURN Best solution found.}
\end{algorithmic}
\end{algorithm}

\vspace{-5mm}
\section{Numerical Results}
\label{Sec:NumRe}
This section provides numerical results to demonstrate the system's performance and the effectiveness of two proposed algorithms in comparison to state-of-the-art benchmarks.  We simulate a network involving up to $50$ APs and $70$ users. These users are uniformly dispersed within a $15$ square kilometer area, visually represented within a three-dimensional Cartesian coordinate system $(x,y,z)$, as illustrated in Fig.~\ref{fig:NetworkModel}. An LEO satellite, positioned at the coordinate ($300,350,400$)~[km], is equipped with 100 antennas and features an antenna gain of 26.9~[dBi], contrasting with the 10.0~[dBi] gain of the ground devices \cite{bisognin2015millimeter}. The system operates on a bandwidth of 100~[MHz] with a carrier frequency of 20~[GHz], employing a coherence block of 10000 OFDM subcarriers for data transmission. Each data symbol is transmitted at 20 dBW \cite{van2022space}. The noise figures are set at 6~[dB] and 1.3~[dB] for the APs and satellite, respectively. The parameter settings adhere to 3GPP standards outlined in \cite{azari2022evolution}, which consider both large and small fading effects in the propagation channels, e.g., for a rural area as $\beta_{nk} = H_m + H_k - 8.50 - 20\log_{10}(f_c) - 38.63\log_{10}(r_{nk}) + \eta_{nk},$ where $H_m$ and $H_k$ are represented the antenna gains at AP~$n$ and user~$k$, respectively, whereas $f_c$ is the carrier frequency. The distance between this user and AP~$n$ is denoted as $r_{nk}$, and shadow fading, $\eta_{nk}$, is modelled as a random variation following a log-normal distribution with a standard deviation of 7~[dB]. The large-scale fading coefficient between user~$k$ and the satellite is determined using one of the models proposed in\cite{giordani2020non} as
\begin{equation}\label{eq:B_k_large}
   \beta_k = H + H_k + \tilde{H}_k - 32.45 - 20\log_{10}(f_cr_k)+ \eta_k,
\end{equation}
where $H$ is the receiver antenna gain at the satellite, and its normalized beam pattern is \cite{bisognin2015millimeter}, as 
\begin{equation}
    \tilde{H}_k = \begin{cases}
4 \left| J_1 \left(\frac{2\pi}{\lambda} \alpha \sin(\upphi_k) \right) / \left(\frac{2\pi}{\lambda} \alpha \sin(\upphi_k) \right) \right|^2, &\text{if } 0 \le \upphi_k \le \frac{\pi}{2}, \\
0, & \text{if } \upphi_k = 0,
\end{cases}
\end{equation}
where $\alpha$ denotes the radius of the antenna’s circular aperture; $\lambda$ is the wavelength; and $\upphi_k$ is the angle between user~$k$ and its beam center. In \eqref{eq:B_k_large}, the shadow fading $\eta_k$ is determined by a log-normal
distribution with the standard derivation depending on the carrier frequency, channel condition, and the elevation angle \cite{giordani2020non}. Besides, $r_k$~[m] represents the distance between the satellite and user~$k$, defined as
$z_k = \sqrt{R^2\sin^2(\theta_k) + z_0^2 + 2z_0R - R\sin(\upphi_k)},$ where $R$ is the Earth’s radius and $z_0$ is the satellite altitude. The system model and proposed algorithms were evaluated through a MATLAB environment installed in a personal desktop with an Intel(R) Core(TM) i7-12700 2.80~[GHz] and 64~[GB] RAM. The following benchmarks compare the performances of the proposed algorithms (Exhaustive search in Algorithm \ref{Alg1} and Binary-coded genetic algorithm in Algorithm \ref{Alg:Integer-coded}) to the full connected model in \cite{van2022space}, another evolutionary algorithms called differential evolutionary (DE) built in \cite{van2023phase} and real-coded genetic algorithm (RCGA) built-in \cite{deb2000efficient}.


\begin{figure}[t]
    \centering
    \subfloat[Arithmetic Mean Throughput]{\includegraphics[width=0.22\textwidth]{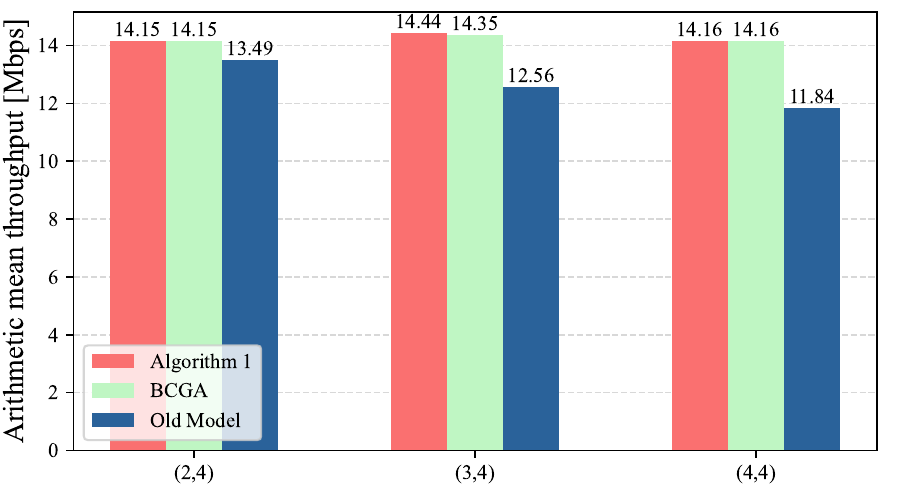}}
    \subfloat[Geometric Mean Throughput]
    {\includegraphics[width=0.22\textwidth]{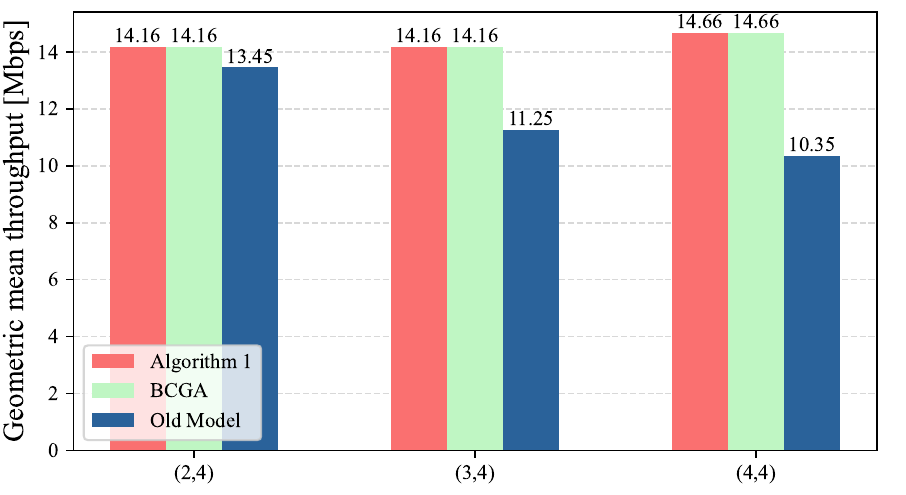}}\\[1ex]
    \subfloat[Max-min Fairness Throughput]{\includegraphics[width=0.22\textwidth]{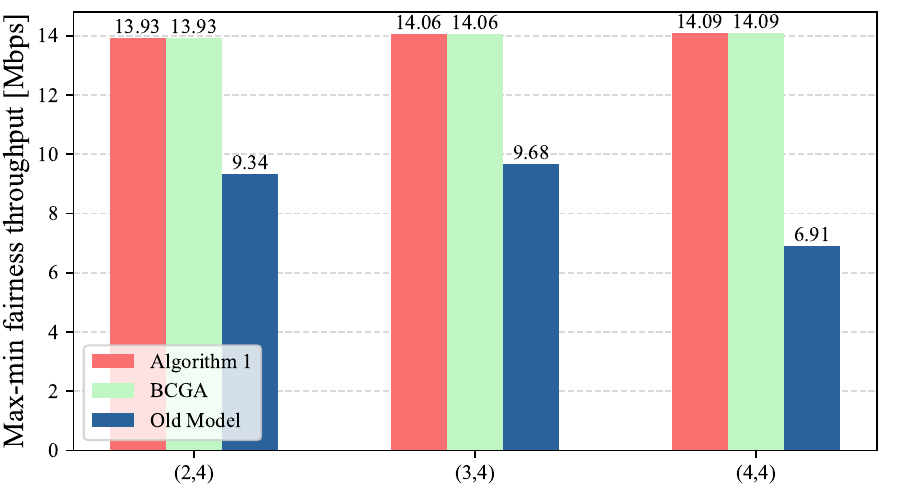}}
    \caption{The optimal results of the proposed algorithms compared the fully connected model.}\vspace{-5mm}
    \label{fig:small_scale}
\end{figure}



\begin{figure*}[!t]
    \centering
    \subfloat[]{\includegraphics[width=0.332\textwidth]{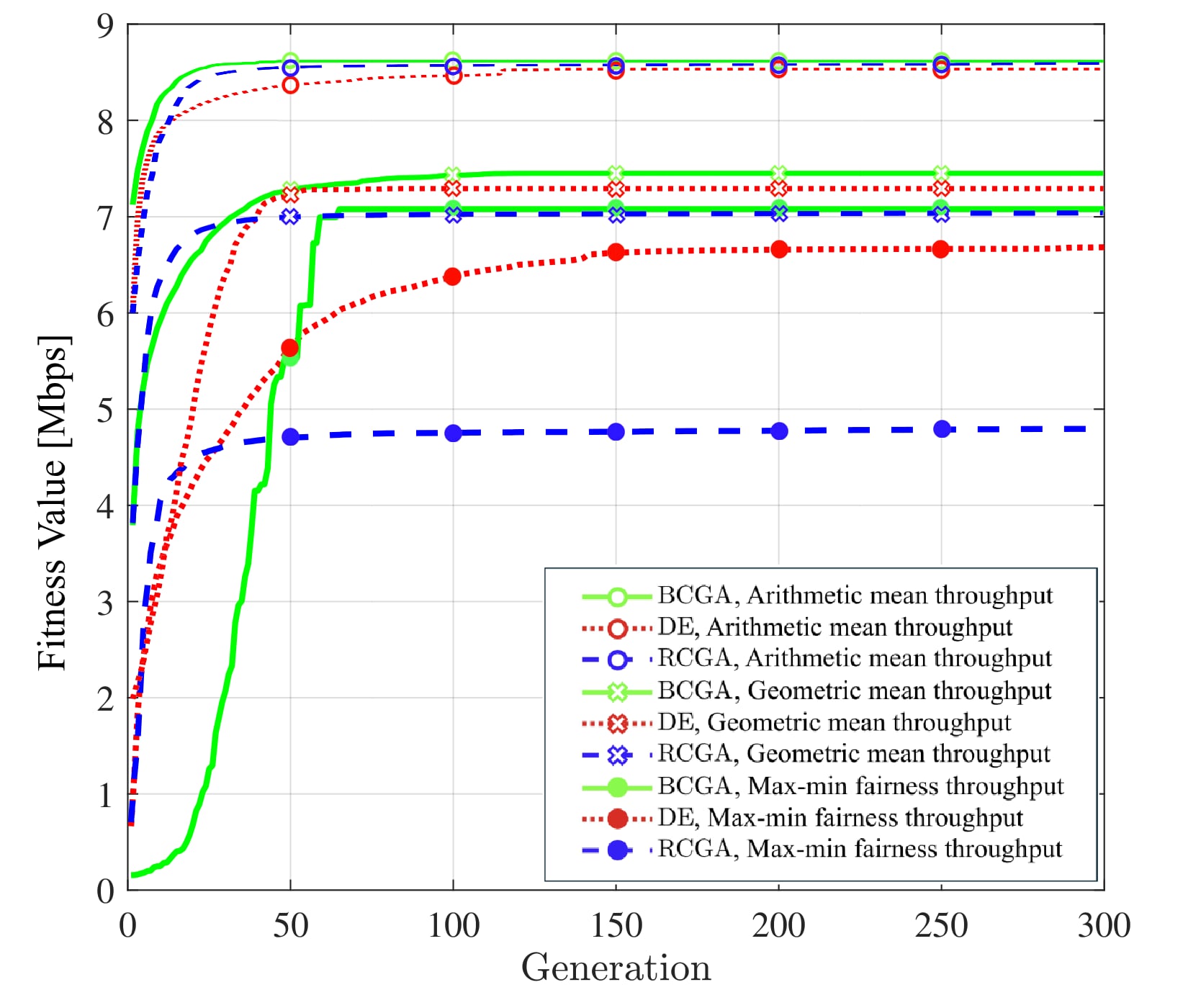}}
    \subfloat[]{\includegraphics[width=0.339\textwidth]{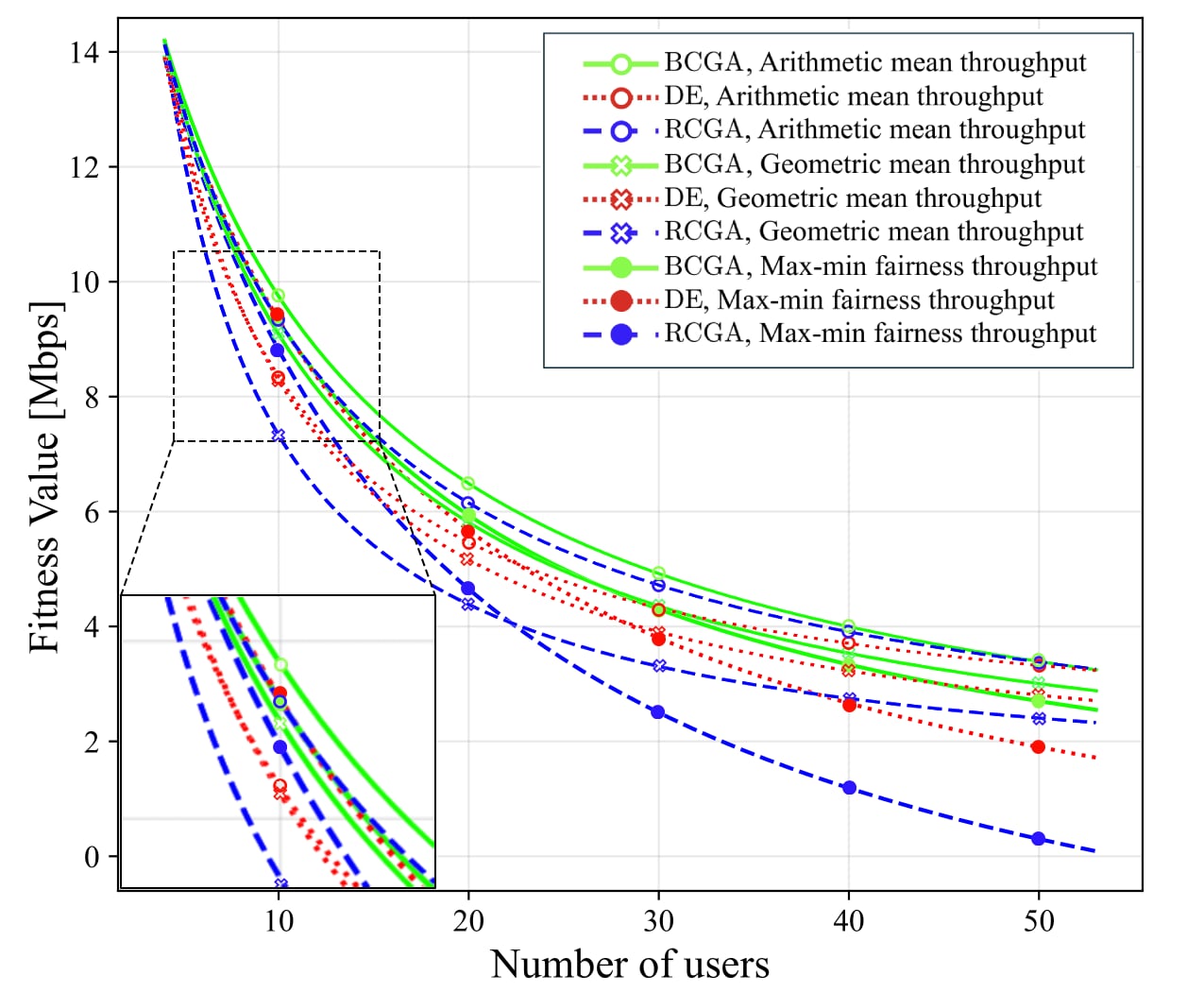}}
    \subfloat[]{\includegraphics[width=0.305\textwidth]{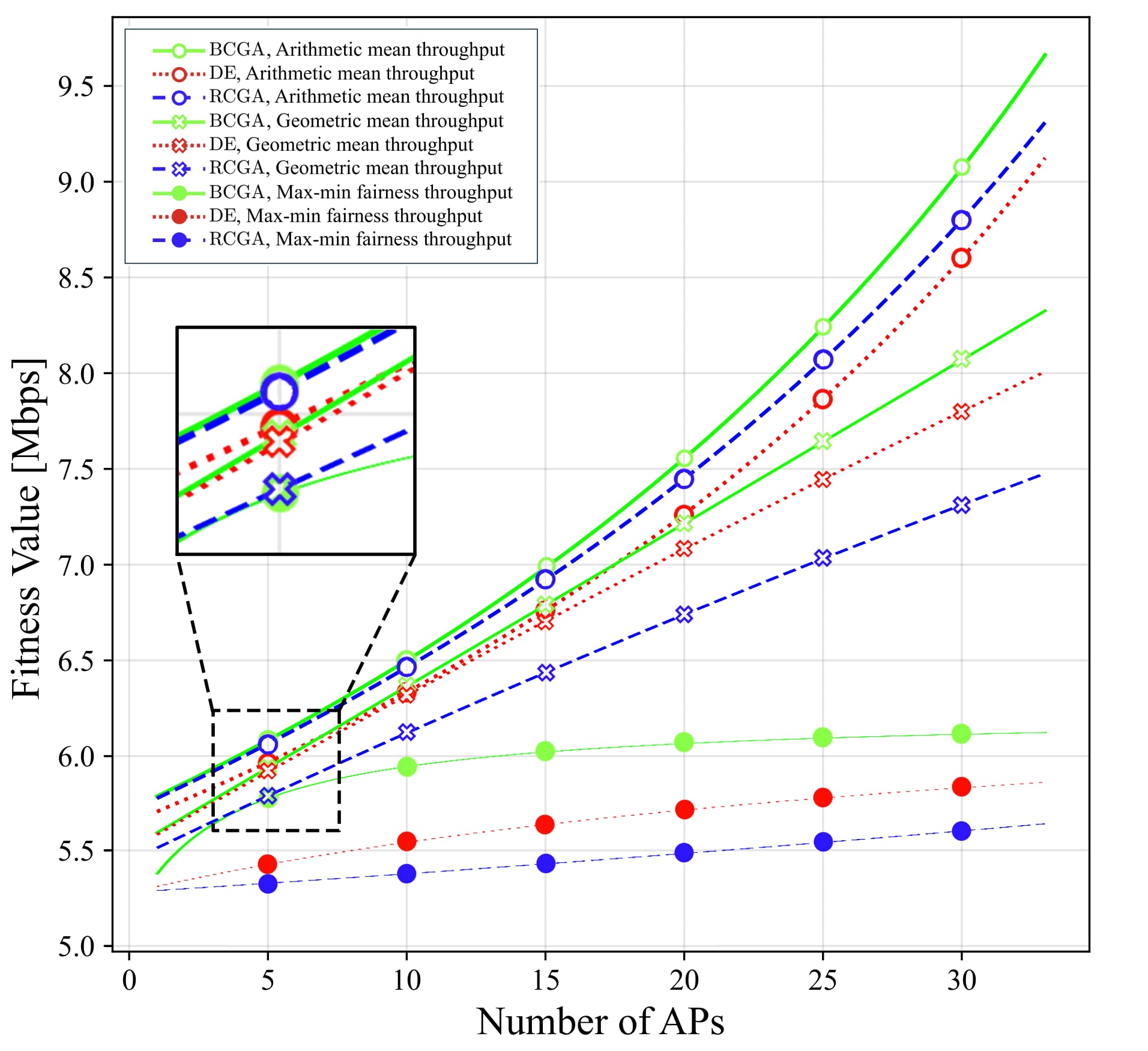}}
    \caption{The system performances of different evolutionary algorithms within model in Problem (19): (a) The convergence rate as the number of generations increases; (b) The evaluation of the gap between the optimal values of different algorithms as the test sets vary; (c) The evaluation of the gap between the optimal values of different algorithms as the test sets vary.}\vspace{-5mm}
    \label{fig:Problem19}
\end{figure*}

In a small-scale network with four users and the numbers of APs is in \{2, 3, 4\} together with the values of fitness functions under different benchmarks are presented in Fig.~\ref{fig:small_scale}. The results demonstrate that Algorithm~\ref{Alg:Integer-coded} yields solutions that are nearly identical to the globally optimal solution obtained through an exhaustive search algorithm \ref{Alg1}. Additionally, we compare the performance of our proposed algorithms with the full association approach. When considering the Arithmetic mean throughput, our proposed algorithms consistently outperform the full association, with fitness values surpassing it by 10\% to 20\%. The superiority of our proposed algorithms becomes even more evident when evaluating the Geometric mean throughput and the Max-min fairness throughput, where improvements range from 30\% to an impressive 50\%.

To observe the convergence of the proposed algorithm (BCGA) and the impact of the problem size, we compare its performance with DE and RCGA. Firstly, within a communication system of $15$ users and $15$ APs, the performance of all these considered benchmarks is evaluated in terms of the number of generations, as displayed in Fig~4a. Across all three throughput metrics, Arithmetic mean throughput, Geometric mean throughput, and Max-min fairness throughput, the BCGA algorithm consistently outperforms both DE and RCGA. In the first utility, arithmetic mean, BCGA converges to a higher value compared to DE and RCGA, showcasing its superiority in maximizing the overall system throughput. In detail, the optimal result provided by BCGA is slightly better than that by DE and RCGA. Meanwhile, for the remaining utilities, geometric mean and max-min fairness, BCGA excels over another counterpart by reaching notable optimal values. These results highlight the potential of BCGA in effectively addressing various utility functions and its robustness in delivering improved network performance compared to other evolutionary algorithms such as DE and RCGA.
Furthermore, while keeping the number of APs at 5 and gradually increasing the number of users, we observe a decreasing trend in the throughput of the objective functions, as shown in Fig.~4b, clearly demonstrating the impact of user density on system performance. It is explained that when the number of users is sufficient for the number of APs, the system can still provide sufficiently good throughput for each objective function; however, as the number of users becomes too large, the throughput of each objective function gradually decreases. This trend is attributed to the fixed serving resources of the system. To investigate this phenomenon further, we analyze the results of the three algorithms for each utility. With the Arithmetic Mean, there is no significant difference among the three algorithms when the number of users exceeds 40. For the Geometric Mean and Max-min fairness utilities, even when the number of users surpasses 50, there remains a small discrepancy between the algorithms. Nevertheless, BCGA consistently delivers the best results across all scenarios. On the contrary, a monotonically increasing trend is observed by fixing the number of users at 20 and continuously increasing the number of APs. Regarding the rate of improvement, for the Arithmetic Mean and Geometric Mean, the optimal value achieved by the proposed algorithms exhibits a significant enhancement compared to smaller values of $N$. Conversely, for the Max-min fairness objective, the optimal value delivered by the algorithms is exhibited as a slight climb. The monotonically increasing trend is a direct consequence of the independent nature of the connections between APs and users. Adding more APs can improve the data throughput of each individual user.

\begin{table}[!t]
    \centering
    \resizebox{0.49\textwidth}{!}{
    \begin{tabular}{lccc}
\hline
\cline{1-4}\rule{0pt}{2.5ex} &\textbf{Arithmetic Utility} & \textbf{Geometric Utility} & \textbf{Max-min Utility} \\ \hline
\rule{0pt}{2.5ex} \textbf{Total Throughput}     & \textbf{231.145}  & 225.927 &   225.365 \\
\rule{0pt}{2.5ex} \textbf{Minimum Throughput}   & 0.0085            & 5.1648  & \textbf{6.3699}\\ \hline \cline{1-4}
\end{tabular}
}
    \caption{The total throughput and minimum throughput in optimized utility systems.}
    \label{tab:total and minimal throughput}
\end{table} 

\begin{table}[!t]
\resizebox{0.48\textwidth}{!}{
\begin{tabular}{lccc}
\hline\cline{1-4}
\rule{0pt}{2.5ex} &
  {\textbf{Arithmetic Utility}} &
  {\textbf{Geometric Utility}} &
  {\textbf{Max-min Utility}} \\ \hline
\rule{0pt}{2.5ex} \textbf{Satellite only} & {56.80}   & {68.80}     & \textbf{74.40}     \\
\rule{0pt}{2.5ex} \textbf{APs only}       & \textbf{40.40}   & {25.20} & {22.40} \\
\rule{0pt}{2.5ex} \textbf{Both of Satellite and APs}  &   {1.20} &   \textbf{6.00} &   {3.20} \\ \hline\cline{1-4}
\end{tabular}
}
\caption{The proportion (\%) of user connections between satellite and APs for different optimization utilities.}\vspace{-6mm}
\label{tab:percent-connnecting}
\end{table}

On the other hand, to evaluate the effectiveness of the benchmarks, Table~\ref{tab:total and minimal throughput} presents the total throughput and minimum throughput for the different fitness functions using through tests problem set of 70 users and 50 APs. It compares the system-level throughput and user-level fairness across different optimization approaches. The total throughput indicates the overall performance of the system. The minimum throughput represents the throughput experienced by the worst-off user, which reflects fairness. In the Arithmetic Throughput utility, the overall system performance reaches its peak with a total throughput of 231.145 [Mbps]. However, some users are neglected, with throughput as low as 0.0085 [Mbps]. Conversely, in the Max-min fairness utility, the performance of each individual user improves significantly, even though the overall system performance is not as high, with a total throughput of 225.365 [Mbps] and a minimum throughput of 6.3699 [Mbps].
The geometric mean throughput utility offers a balanced trade-off between these two utilities, as previously mentioned. It achieves a total throughput of 225.927 [Mbps], which is slightly lower than the arithmetic throughput utility but higher than the max-min fairness utility. Additionally, the minimum throughput in the Geometric Mean Throughput utility is 5.1648 [Mbps], which is a significant improvement compared to the Arithmetic Throughput utility and only slightly lower than the Max-min fairness utility. Therefore, the Geometric Mean Throughput utility balances maximizing total system performance and ensuring fairness among users. This makes it a proper optimization target that considers the demands of all users.

The roles of the space and ground links in enhancing the uplink spectral efficiency are captured by optimizing the connection of the satellite, APs, and users. We compared the differences between the ergodic throughput of three distinct network architectures: spacial satellite-only (Mode S), terrestrial APs only (Mode A), and our proposed cell-free-terrestrial integrated model in Fig.~\ref{fig:compare_model}. The box plots demonstrate that our proposed model achieves superior performance metrics compared to both baseline approaches. Using only terrestrial APs might reduce throughput in areas with high user density. Similarly, satellite-only deployment yields reduced user throughput, albeit with minimal variance observed across experimental trials. This indicates a hybrid approach, providing users with the benefits of the joint connectivity where both APs and satellite connections are optimized for different types of traffic or service levels. 
Moreover, our analysis of the model optimized by Algorithm~\ref{Alg:Integer-coded} reveals preferred connections to either satellite or APs, with priority given to the satellite due to the LoS links.
Table~\ref{tab:percent-connnecting} shows the percentage of users served only from either the single link or from both. Furthermore, the preference for satellite connectivity is most pronounced in the max-min fairness that optimizes the connection quality for the weakest users. Specifically, this objective shows the highest percentage of users up to 74.40\% served only by the satellite, suggesting a prioritization of ensuring that each user achieves at least a minimum level of service quality, even if it might reduce the overall system performance. This observation highlights reliance on satellite connections, which may provide more consistent coverage and bandwidth compared to terrestrial networks. Integrating a satellite enhances coverage, especially in areas where terrestrial APs struggle. Without the satellite, terrestrial networks would offer limited coverage. Conversely, the sum throughput utility shows a percentage of users (56.80\%) served by the satellite only. This objective distributes resources more evenly across all the users, balancing the load between the satellite and APs. In all the realizations, there are 40.40\% served by the only APs. This approach can lead to a compromise in individual performance to enhance the total throughput, making it suitable for scenarios where the overall network performance is critical. The geometric throughput, which balances between the total network and per-user throughput, shows intermediate values in both categories, with an enhanced percentage of users about 6.00\% connecting through both the satellite and APs.

\textcolor{black}{We benchmark the system performance with and without the incorporation of power efficiency constraints across three utility optimization models in Fig.~\ref{fig:compare_model_withenergy}. The numerical evaluations consistently reveal that integrating power efficiency constraints yields a system throughput gain ranging from 10\% to 20\%. These findings confirm that although this enhancement entails a moderate increase in computational complexity, it leads to a notably more energy-efficient and high-performing network design.}

\begin{figure}[t]
    \centering
    \includegraphics[width=.49\textwidth]{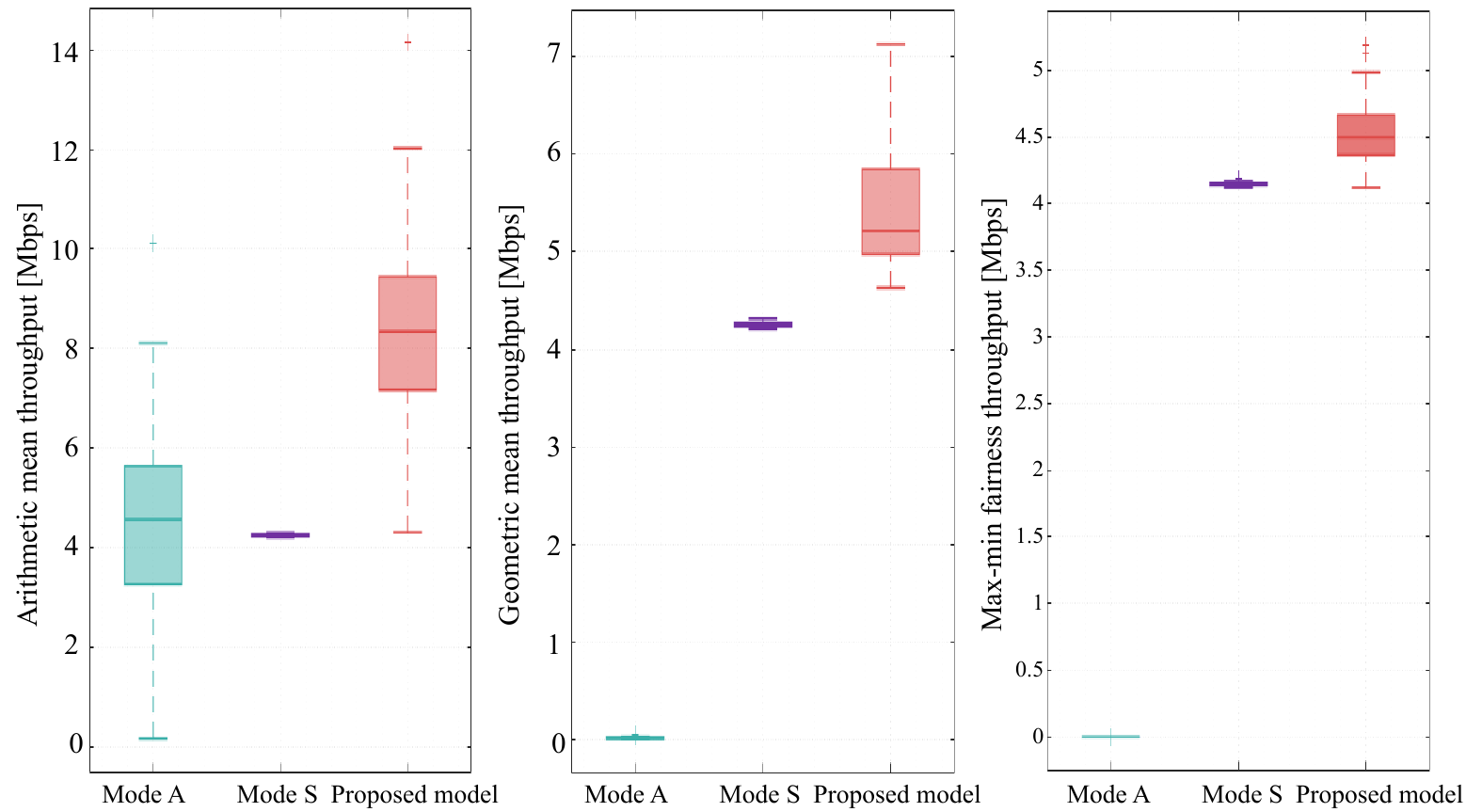}
    \caption{The comparison of ergodic throughput across different network architectures.}\vspace{-5mm}
    \label{fig:compare_model}
\end{figure}

\begin{figure}
    \centering
    \includegraphics[width=0.85\linewidth]{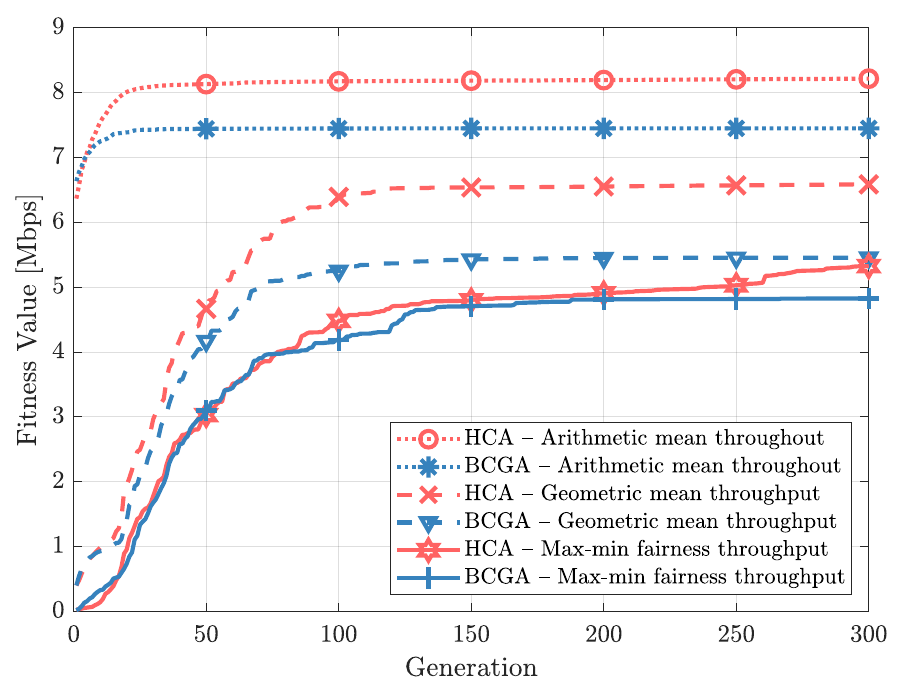}
    \caption{\color{black} The system performance comparison between the fixed and optimized power allocation systems based on the optimal results obtained by BCGA and HGA.}
    \label{fig:compare_model_withenergy}
\end{figure}
\section{Conclusion} \label{Sec:Con}

This paper has investigated the applications of adapted GA to solve the user association in integrated space-ground cell-free Massive MIMO systems, where a load balancing approach is utilized to enhance spectral efficiency. Moreover, we formulated and solved the APs and satellite cooperation problem in an effective manner by adjusting the Binary-coded genetic algorithm. Numerical results demonstrated that the proposed algorithm BCGA yields the fairness level closed to the global optimum for small-scale networks. Besides, it provides a superior solution compared with the full association benchmarks for large-scale networks. We also highlighted the critical roles of the satellite and APs in improving the  throughput of each user. \textcolor{black}{Furthermore, the inclusion of power control strategies into the proposed framework enables more flexible resource allocation, which significantly boosts system performance under varying user demands and channel conditions.}

\vspace{-5mm}
\appendix
\subsection{Proof of Theorem~\ref{Theorem:SE}} \label{Appendix:SE}
The closed-form expression of the desired signal gain is 
\begin{equation} 
\begin{split}
& \mathbb{E}\{o_{kk} \}  = \tilde{\alpha}_{k}  \mathbb{E} \{ \hat{\mathbf{h}}_{k}^H \mathbf{h}_{k} \} + \sum\nolimits_{n=1}^N \tilde{\alpha}_{k} \mathbb{E}\{\hat{g}_{nk}^\ast g_{nk} \} \\
&  \stackrel{(a)}{=} \tilde{\alpha}_{k}  \mathbb{E} \{  \| \hat{\mathbf{h}}_{k} \|^2 \} + \sum\nolimits_{n=1}^N \tilde{\alpha}_{k} \mathbb{E}\{|\hat{g}_{nk}|^2 \} \\
& \stackrel{(b)}{=} \tilde{\alpha}_k \| \overline{\mathbf{h}}_k \|^2 + pK \tilde{\alpha}_k \mathrm{tr}(\mathbf{R}_k \pmb{\Psi}_k \mathbf{R}_k) + \alpha_k\sum\nolimits_{n=1}^N \varrho_{nk},
\end{split}
\end{equation}
where $(a)$ is obtained by the independence of the channel estimate and its estimation error from user~$k$ to the satellite and APs. Meanwhile, $(b)$ is obtained by the channel statistics in Lemma~\ref{lemma:ChannelEst}. Consequently, the numerator of \eqref{eq:SINRk} is
\begin{multline}
p_k |\mathbb{E}\{o_{kk} \}|^2 = \\
p_k \left| \tilde{\alpha}_k \| \overline{\mathbf{h}}_k \|^2 + pK \tilde{\alpha}_k \mathrm{tr}(\mathbf{R}_k \pmb{\Psi}_k \mathbf{R}_k) + \alpha_k\sum\nolimits_{n=1}^N \varrho_{nk}\right|^2.
\end{multline}
The first term in the denominator of \eqref{eq:SINRk} is reformulated as
\begin{equation} \label{eq:SumT1}
\sum\nolimits_{k'=1}^K p_{k'} \mathbb{E}\{ |o_{kk'}|^2 \} = \sum\nolimits_{k'=1, k' \neq k}^K p_{k'} \mathbb{E}\{ |o_{kk'}|^2 \} +   p_{k} \mathbb{E}\{ |o_{kk}|^2 \}.
\end{equation}
Each expectation of the mutual interference in \eqref{eq:SumT1} is computed for $k \neq k'$ as follows
\begin{equation} \label{eq:Expokkprime}
\begin{split}
& \mathbb{E}\{ |o_{kk'}|^2 \} = \tilde{\alpha}_{k'} \mathbb{E} \{ | \hat{\mathbf{h}}_k^{H} \mathbf{h}_{k'} |^2 \} +  \tilde{\alpha}_{k'} \sum\nolimits_{n=1}^N \mathbb{E} \{ |\hat{g}_{nk} g_{nk'}|^2 \} =\\
&  \tilde{\alpha}_{k'} \mathrm{tr} (\mathbf{R}_{k'} \mathbf{R}_{k} \pmb{\Psi}_k \mathbf{R}_{k}) + \tilde{\alpha}_{k'}  \bar{\mathbf{h}}_{k'}^H \mathbf{R}_{k} \pmb{\Psi}_k \mathbf{R}_{k} \bar{\mathbf{h}}_{k'}  + \tilde{\alpha}_{k'} \alpha_k \sum\nolimits_{n=1}^N \varrho_{nk} \beta_{nk'},
\end{split}
\end{equation}
where the first expectation on the right-hand side of \eqref{eq:Expokkprime} is computed in closed form by decomposing the space link into the LoS and NLoS components. The second expectation on the right-hand side of \eqref{eq:Expokkprime} is obtained by the independence of the channel estimate and the estimation error. For $k' = k$, the second expectation on the right-hand side of \eqref{eq:SumT1} is driven as
 
\begin{align}
\label{eq:Eokk}
&\mathbb{E}\{ |o_{kk}|^2 \} = \mathbb{E} \left\{ \left| \| \hat{\mathbf{h}}_k \|^2 + \hat{\mathbf{h}}_k^H \mathbf{e}_k + \sum_{n=1}^N |\hat{g}_{nk}|^2 + \sum_{n=1}^N \hat{g}_{nk}^\ast e_{nk}  \right|^2 \right\} \nonumber\\
&=  \mathbb{E} \{\| \hat{\mathbf{h}}_k \|^4 \} +  \mathbb{E}\{ |\hat{\mathbf{h}}_k^H \mathbf{e}_k|^2  \} +  \mathbb{E} \left\{ \left| \sum\nolimits_{n=1}^N |\hat{g}_{nk}|^2  \right|^2 \right\}  \nonumber\\
& + \mathbb{E} \left\{ \left| \sum\nolimits_{n=1}^N \hat{g}_{nk}^\ast e_{nk}  \right|^2  \right\} + 2 \mathbb{E} \left\{ \| \hat{\mathbf{h}}_k \|^2   \sum\nolimits_{n=1}^N |\hat{g}_{nk}|^2  \right\},
\end{align}
where the remaining expectations in  the last equation of \eqref{eq:Eokk} are disappeared since the zero mean of the additive noise. The first expectation in the last equation of \eqref{eq:Eokk} is computed as
\begin{multline} \label{eq:hatk4}
\mathbb{E} \{\| \hat{\mathbf{h}}_k \|^4 \} = \left( \tilde{\alpha}_k \| \bar{\mathbf{h}}_k  \|^2 + 2pK \tilde{\alpha}_k\mathrm{tr}\left( \mathbf{R}_k \pmb{\Psi}_k \mathbf{R}_k \right) \right)^2 +  \\
 2 pK \tilde{\alpha}_k \bar{\mathbf{h}}_k^H \mathbf{R}_k \pmb{\Psi}_k \mathbf{R}_k \bar{\mathbf{h}}_k  + p^2 K^2 \tilde{\alpha}_k \mathrm{tr}\left( \mathbf{R}_k \pmb{\Psi}_k \mathbf{R}_k  \mathbf{R}_k \pmb{\Psi}_k \mathbf{R}_k \right), 
\end{multline}
which is obtained based on the channel estimate in \eqref{eq:hatk} and the Boolean property $\tilde{\alpha}_k^2 = \tilde{\alpha}_k$ .The second expectation in the last equation of \eqref{eq:Eokk} is computed as
\begin{multline}
\mathbb{E}\{ |\hat{\mathbf{h}}_k^H  \mathbf{e}_k |^2 \} = \tilde{\alpha}_k \bar{\mathbf{h}}_k^H \mathbf{R}_k \bar{\mathbf{h}}_k - p K \tilde{\alpha}_k\bar{\mathbf{h}}_k^H \mathbf{R}_k \pmb{\Psi}_k  \mathbf{R}_k \bar{\mathbf{h}}_k + \\
pK \tilde{\alpha}_k \mathrm{tr} \left( \mathbf{R}_k \mathbf{R}_k \pmb{\Psi}_k \mathbf{R}_k  \right) - p^2K^2 \tilde{\alpha}_k \mathrm{tr} \left( \mathbf{R}_k \pmb{\Psi}_k \mathbf{R}_k  \mathbf{R}_k \pmb{\Psi}_k \mathbf{R}_k  \right),
\end{multline}
where the channel estimate is given in \eqref{eq:hatk} and  the channel estimation error together with its statistics are defined in Lemma~\ref{lemma:ChannelEst}. The third expectation in the last equation of \eqref{eq:Eokk} is computed as
\begin{equation}
\mathbb{E} \left\{ \left| \sum\nolimits_{n=1}^N |\hat{g}_{nk}|^2  \right|^2 \right\} = \sum\nolimits_{n=1}^N \varrho_{nk}^2  + \left( \sum\nolimits_{n=1}^N \varrho_{nk}  \right)^2,
\end{equation}
where the channel estimate $\hat{g}_{nk}$ is defined in \eqref{eq:hatgnk} and its variance given in \eqref{eq:varrhonk}. The fourth expectation in the last equation of \eqref{eq:Eokk} is computed as
\begin{equation}
\mathbb{E} \left\{ \left| \sum\nolimits_{n=1}^N \hat{g}_{nk}^\ast e_{nk}  \right|^2  \right\} = \sum\nolimits_{n=1}^N \varrho_{nk} (\beta_{nk} - \varrho_{nk}),
\end{equation}
which is based on the independence of $\hat{g}_{nk}$ and $e_{nk}, \forall n,k$. The last expectation in the last equation of \eqref{eq:Eokk} is computed as
\begin{multline} \label{eq:hatgk}
\mathbb{E} \left\{ \| \hat{\mathbf{h}}_k \|^2   \sum\nolimits_{n=1}^N |\hat{g}_{nk}|^2  \right\} = \\
\left( pK \tilde{\alpha}_k \mathrm{tr}(\mathbf{R}_k \pmb{\Psi}_k \mathbf{R}_k ) + \tilde{\alpha}_k \| \bar{\mathbf{h}}_k \|^2  \right) \sum\nolimits_{n=1}^N \varrho_{nk},
\end{multline}
thanks to the independent of  the space and ground channels. Substituting, \eqref{eq:hatk4}-\eqref{eq:hatgk} into \eqref{eq:Eokk}, obtaining the closed-form expression
\begin{multline}
\mathbb{E}\{ |o_{kk}|^2 \} = \left(  \tilde{\alpha}_k \|\bar{\mathbf{h}}_k \|^2 +  pK \tilde{\alpha}_k\mathrm{tr}(\mathbf{R}_k \pmb{\Psi}_k \mathbf{R}_k) \right)^2 + \\
2 pK \tilde{\alpha}_k \bar{\mathbf{h}}_k^H \pmb{\Psi}_k \bar{\mathbf{h}}_k + p^2 K^2 \tilde{\alpha}_k  k\mathrm{tr}(\mathbf{R}_k \pmb{\Psi}_k \mathbf{R}_k \mathbf{R}_k \pmb{\Psi}_k \mathbf{R}_k),
\end{multline}
after some algebra. The other expectations are computed in a similar manner and we obtain the result as  in the theorem.
\vspace{-5mm}
\subsection{Useful Definition and Lemmas of Theorem \ref{Theorem:EFHT}}
\begin{definition}[\cite{he2001drift},\cite{he2003towards}]\label{Def:1}
{In the initial stage, we shall introduce essential notations and definitions used for the proof. Let us consider a population space denoted by $\mathcal{X}$, and  $\mathcal{X}^\ast$ represent a collection of all the optimal solutions. We define a Markov chain as a sequence $\{{\zeta_t}\}^{\infty}_{t=0}$, where each $\zeta_t$ is a state at time instance $t$ ($t= 0,1,\ldots$). The notation $\mu_t$ denotes the probability of $\zeta_t$ belonging to $\mathcal{X}^\ast$, which is
\begin{equation}
\mu_t = \sum\nolimits_{x \in X^{*}} \mathrm{Pr}(\zeta_t = x),
\end{equation}
where $\mathrm{Pr}(\cdot)$ is the probability of an event. We note that $\{{\zeta_t}\}^{\infty}_{t=0}$ is said to converge to $\mathcal{X}^{\ast}$ if $\lim_{t \to \infty} \mu_t = 1$ and the convergence rate is measured by $1 - \mu_t$ at time instance $t$. In this paper, absorbing the Markov chain will be utilized to model the procedure that obtain a solution to problem~\eqref{Problem:MaxMinQoS} due to its desirable theoretical attributes and their feasibility in practical applications. We recall that $\{{\zeta_t}\}^{\infty}_{t=0}$ is an absorbing chain if
  $\forall t \in \{0,1 ,\ldots\}:  \mathrm{Pr}(\zeta_{t+1} \notin \mathcal{X}^{\ast} \mid \zeta_t \in \mathcal{X}^{\ast}) = 0.$
Let's introduce a random variable $\tau$  representing the events
\begin{multline}
\tau = 0: \zeta_0 \in \mathcal{X}^{\ast} \mbox{ and } \tau = t: \zeta_t \in \mathcal{X}^{\ast} \wedge \zeta_i \notin \mathcal{X}^{\ast} \\ (\forall i \in \{0, 1, \ldots, t-1\}), \forall t \geq 1. 
\end{multline}
The expectation $\mathbb{E}\{\tau\}$, is called the expected first hitting time of the above Markov chain.} 
\end{definition}

\begin{lemma}[\cite{yu2008new}] \label{lemma2}
{Given an absorbing Markov chain $\{{\zeta_t}\}^{\infty}_{t=0}$ with $\zeta_t \in \mathcal{X}$ and a target subspace $\mathcal{X}^{\ast} \subset \mathcal{X}$, if two sequences $\{A_t\}_{t = 0}^{\infty}$ and $\{B_t\}_{t = 0}^{\infty}$ satisfy
\begin{align}
&\prod\nolimits_{t = 0}^{\infty}(1-A_t) = 0, \\
&  B_t \geq \sum\nolimits_{x \notin \mathcal{X}^\ast} P(\zeta_{t+1} \in \mathcal{X}^{\ast} \mid \zeta_t = x)\frac{P(\zeta_t = x)}{1-\mu_t} \geq A_t, \label{eq:beta}
\end{align}
then the Markov chain converges to $\mathcal{X}^{\ast}$ and the convergence rate $(1 - \mu_t)$ is bounded by
\begin{equation}
    (1 - \mu_0)\prod\nolimits_{i=0}^{t-1}(1-A_i) \geq 1 - \mu_t \geq (1-\mu_0)\prod\nolimits_{i=0}^{t-1}(1-B_i).
\end{equation}}
\end{lemma}
\begin{lemma}[\cite{yu2008new}] \label{lemma3}
{Let $m$ and $n$ represent two discrete random variables that take on non-negative integer values with the finite expectations. We denote $F_m(\cdot)$ and $F_n(\cdot)$  their respective cumulative distribution functions
\begin{align}
& F_m(t)=\mathrm{Pr}(m \leq t)=\sum\nolimits_{i=0}^t \mathrm{Pr}(m=i), \\ 
&  F_n(t)= \mathrm{Pr}(n \leq t)=\sum\nolimits_{i=0}^t \mathrm{Pr}(n=i).
\end{align}
If $F_m(t) \geq F_n(t)(\forall t=0,1, \ldots)$, then the expectations  satisfy 
\begin{equation}
\mathbb{E}\{m \} \leq \mathbb{E} \{ n \},
\end{equation}
with $\mathbb{E}\{u\}=\sum_{t} t \mathrm{Pr}(u=t),  \mathbb{E}\{v \} =\sum_{t} t \mathrm{Pr}(v=t).$}
\end{lemma}
\subsection{Proof of Theorem~\ref{Theorem:EFHT}} \label{Appendix:EFHT}
{By utilizing Lemma~\ref{lemma2} with \eqref{eq:beta}, one obtains
\begin{equation}
    1-\mu_t \leq \left(1-\mu_0\right) \prod\nolimits_{i=0}^{t-1}\left(1-A_i\right).
\end{equation}
Note that  $\mu_t$ expresses the distribution of $\tau$, i.e., $\mu_t=F_\tau(t)$, we can get the lower bound of $F_\tau(t)$ as
\begin{equation}
    F_\tau(t) \geqslant \begin{cases}
    \mu_0, & t=0, \\ 1-\left(1-\mu_0\right) \prod_{i=0}^{t-1}\left(1-A_i\right), & t \geq 1. 
    \end{cases}
\end{equation}
Let's denote a virtual random variable $\eta$ with its cumulative distribution function equal the lower bound of $F_\tau(t)$. After that, the expectation of $\eta$ is
\begin{equation}
    \begin{aligned}
&\mathbb{E}\{\eta\}  =0 \mu_0+1 \left(1-\left(1-A_0\right)\left(1-\mu_0\right)-\mu_0\right) + \\
&\sum\nolimits_{t=2}^{+\infty} t \left(\left(1-\mu_0\right) \prod\nolimits_{i=0}^{t-2}\left(1-A_i\right)-\left(1-\mu_0\right) \prod\nolimits_{i=0}^{t-1}\left(1-A_i\right)\right) \\
& =\left(A_0+\sum\nolimits_{t=2}^{+\infty} t A_{t-1} \prod\nolimits_{i=0}^{t-2}\left(1-A_i\right)\right)\left(1-\mu_0\right).
\end{aligned}
\end{equation}
Since $D_\tau(t) \geq D_\eta(t)$, according to Lemma~\ref{lemma3}, $\mathbb{E}\{ \tau \} \leq \mathbb{E}\{\eta\}$. Thus, the upper bound of the expected first hitting time is
\begin{equation}
  \mathbb{E} \{ \tau \} \leq\left(A_0+\sum\nolimits_{t=2}^{+\infty} t A_{t-1} \prod\nolimits_{i=0}^{t-2}\left(1-A_i\right)\right)\left(1-\mu_0\right).
\end{equation}
We have assumed that  Algorithm~\ref{Alg:Integer-coded} begins with a non-optimal solution, which implies $\mu_0 =0$. In a similar manner, the lower bound of the expected first hitting time can be derived as
\begin{equation}\label{eq:lowerbound}
  \mathbb{E}\{\tau\} \geq \left(B_0+\sum\nolimits_{t=2}^{+\infty} t B_{t-1} \prod\nolimits_{i=0}^{t-2}\left(1-B_i\right)\right)\left(1-\mu_0\right).
\end{equation}
By mutation, the max probability of a solution being mutated to be the optimal solution is $p_m(1-p_m)^{2K-1}$. Consequently, the maximum probability of a mutated population to become an optimal population is $1-(1-p_m(1-p_m)^{2K-1})^{2K}$. Due to the fact that
\begin{equation}
\mathrm{Pr}(\zeta_{n+1} \in \mathcal{X}^{\ast} \mid \zeta_t =  x  ) \leq 1-(1-p_m(1-p_m)^{2K-1})^{2K},
\end{equation}
we obtain the following upper bound
     \begin{equation}\label{eq:upperbound}
      \begin{split}
         &\sum\nolimits_{ x \notin \mathcal{X}^\ast}  \mathrm{Pr}(\zeta_{n+1}\in X^*\mid \zeta_t= x) \frac{\mathrm{Pr}(\zeta_t= x)}{1-\mu_t} \\
         &\leq \sum\nolimits_{ x \notin \mathcal{X}^\ast} (1-(1-p_m(1-p_m)^{2K-1})^{2K})\frac{\mathrm{Pr}(\zeta_t= x )}{1-\mu_t}\\
         &=(1-(1-p_m(1-p_m)^{2K-1})^{2K})\frac{\sum_{ x \notin \mathcal{X}^{\ast}} \mathrm{Pr}(\zeta_t= x )}{1-\mu_t}\\
         &\stackrel{(a)}{\approx}1-(1-p_m(1-p_m)^{2K-1})^{2K},
   \end{split}      
    \end{equation}
where $(a)$ is obtained by the definition $\mu_t = \sum_{ x \in \mathcal{X}^\ast} \mathrm{Pr}(\zeta_t=x)$ as shown in Definition~\ref{Def:1}. Hence, $1- \mu_t = \sum_{ x \notin \mathcal{X}^\ast} \mathrm{Pr}(\zeta_t=x)$. Let $ B_t=2Kp_m(1-p_m)^{2K-1}$ and by utilizing \eqref{eq:lowerbound}, we obtain
\begin{align}
     &\mathbb{E}\{\tau\} \geq B_o+\sum\nolimits_{t=2}^{\infty} tB_{t-1} \prod\nolimits_{i=0}^{t-2}(1-B_i)=\frac{1}{n}\frac{1}{p_m}\left( \frac{1}{1-p_m}\right)^{2K-1}\nonumber\\
     &=\frac{(1-p_m)^{-2K}}{2Kp_m}.
\end{align}
In fact, $p_m$ is predetermined and therefore, $\mathbb{E}\{\tau\}$ is obtained as in the theorem.}
\vspace{-5mm}
\subsection{Proof of Theorem~\ref{Theorem:EFHT_HGA}}\label{Appendix:Proof_EFHTofHGA}
\textcolor{black}{We extend Theorem~\eqref{Theorem:EFHT} by modeling HGA as a Markov chain over the hybrid state space. Each solution is a vector $\pmb{\gamma} = (\pmb{\phi}, \pmb{\xi})$, where $i)$ $\pmb{\phi} \in \{0,1\}^{2K}$ (binary associations), $ii)$ $\pmb{\xi} \in \{0, \frac{1}{L}, \frac{2}{L}, \dots, 1\}^K$ (discretized power levels). The entire state space has size $2^{2K} \times (L+1)^K$. Let analysis the mutation probability: 1. Probability of binary-mutating to the optimal $\pmb{\phi}^*$ is $\leq 1- (1-p_m)^{2K}$ (as in Theorem~\ref{Theorem:EFHT}). 2. Probability of Real-mutating to the optimal $\pmb{\xi}^*$ is $\leq 1- (1-\eta_m)^K$, where $\eta_m$ is the minimal probability to jump to a specific discretized level. Real-valued power variables $\xi_k$ are discretized into $L$ levels (e.g., fixed-point precision). Then, the minimal mutation probability to reach the optimal discretized power level is $\eta_m = \Theta(1/L)$.
In worst-case transition probability, the probability to generate the optimal solution from any suboptimal state is bounded below by $\Pr(p_{\text{hybrid}}) \leq 1- \underbrace{(1-p_m)^{2K}}_{\text{binary}} \times \underbrace{(1-\eta_m)^K}_{\text{real}}.$ Using the same Markov chain absorption and framework as \eqref{eq:upperbound} and \eqref{eq:lowerbound}, the EFHT satisfies 
\begin{equation}
\mathbb{E}\{\tau\} \geq\Omega\left( \frac{1}{(1-p_m)^{2K} (1-\eta_m)^K |\mathcal{S}|} \right).
\end{equation}
Accounting for $|\mathcal{S}| = \Theta(K)$ (as in Theorem~\ref{Theorem:EFHT}), we obtain the bound $\mathbb{E}\{\tau\} \geq \tilde{c}  (1-p_m)^{-2K} (1-\eta_m)^{-K} K^{-1}$.}
\vspace{-3mm}

\bibliographystyle{IEEEtran}
\bibliography{refs}

\end{document}